\documentclass[10pt]{amsart}

\usepackage[english]{babel}

\usepackage{amsmath}
\usepackage{graphicx}
\graphicspath{{fig/}}

\usepackage{tikz}
\usepackage[latin1]{inputenc}
\usepackage[T1]{fontenc}
\usepackage{verbatim,url}
\usepackage{enumitem}
\usepackage{amsfonts,amsmath,amscd,epsfig,bbm,subfigure,graphics}

\usepackage[pdfpagelabels,plainpages=false,colorlinks=true,%
  colorlinks=true,%
  linkcolor=black!70!red!100,%
  urlcolor=black!70!blue!100,%
  citecolor=black!70!green!100]{hyperref}

\newcommand\R{\mathbb{R}}

\newcommand{\floor}[1]{\left\lfloor#1\right\rfloor}

\newcommand{\lp}{\left(}

\newcommand{\rp}{\right)}

\newif\ifRapportRecherche
\RapportRecherchefalse 
\RapportRecherchetrue  

\ifRapportRecherche

\else

\fi

\def\R{{\mathbb R}}

\title[Incentives and Redistribution in Bike-Sharing Systems]{Incentives and Redistribution in Homogeneous Bike-Sharing Systems with
  Stations of Finite Capacity}
\author{Christine Fricker \and Nicolas Gast}
\thanks{This work is partially supported by the EU project QUANTICOL, 600708.}
\thanks{C. Fricker is with INRIA Paris-Rocquencourt Domaine de
  Voluceau, 78153 Le Chesnay, France.
  \url{christine.fricker@inria.fr}%
  \and %
  Nicolas Gast is with IC-LCA2, EPFL, Lausanne,
  Switzerland. and Inria, Grenoble, France. \url{nicolas.gast@inria.fr} }

\date{\today}

\newtheorem{theorem}{Theorem}

\begin{document}

\maketitle

\begin{abstract}
  Bike-sharing systems are becoming important for urban transportation. In such systems, users arrive at a station, take a bike and use it for a while, then return it to another station of their choice. Each station has a finite capacity: it cannot host more bikes than its capacity.  We propose a stochastic model of an homogeneous bike-sharing system and study the effect of users random choices on the number of problematic stations, \emph{i.e.}, stations that, at a given time, have no bikes available or no available spots for bikes to be returned to. We quantify the influence of the station capacities, and  we compute the fleet size that is optimal in terms of minimizing the proportion of problematic stations. Even in a homogeneous city, the system exhibits a poor performance: the  minimal proportion of problematic stations is of the order of (but not lower than) the inverse of the capacity.  We show that simple incentives, such as suggesting users to return to the least loaded station among two  stations, improve the situation by an exponential factor.  We also compute the rate at which bikes have to be redistributed by trucks to insure a given quality of service. This rate is of the order of the inverse of the station capacity. For all cases considered, the fleet size that corresponds to the best performance is half of the total number of spots plus a few more, the value of the \emph{few more} can be computed in closed-form as a function of the system parameters. It corresponds to the average number of bikes in circulation.
\end{abstract}


\section{Introduction}
Bike-sharing systems (BSS) are becoming important for urban transportation. They are devoted to short trips.  A few BSS have been launched
since Copenhagen launched its program in 1995.  BSS were widely deployed in
the 2000s after Paris launched the large-scale program called Velib,
in July 2007.  Velib consists of 20\,000 available bikes and 1500
stations. Currently, there are more than $400$ cities equipped with BSS
around the world (see \cite{demaio2009bike} for a history of BSS).
The popularity of BSS gives rise to a recent research
activity. 

The concept of BSS is simple: A user arrives at a station, takes a
bike, uses it for a while and then returns it to another station. The
lack of resources is one of the major issues: a user can arrive at a
station that has no bike available, or wants to return her bike at a
station with no empty spot.  The allocation of resources, bikes and
empty places, has to be managed by the operator in order to offer a
reliable alternative to other transportation modes.

The strength of such a system is its ability to meet the demand, in
bikes and empty spots. This demand is complex. It depends on the time of the
day, the day of the week (week or week-end), the season and the
weather, but also the location: housing or working areas generate
going-and-coming flows; flows are also generated from up-hill to
down-hill stations. This creates unbalanced traffic during the
day. Moreover, the system is stochastic due to the arrivals at the
stations, the origin-destination pairs and the trip lengths. The lack
of resources also generates random choices from the users, who must
search for another station.  These facts are supported by several data
analyses, \emph{e.g.}, \cite{Borgnat-1,Froehlich-1,Nair-1}.

When building a bike-sharing system, a first strategic decision is the
planning of the number of stations, their locations and their
size. Other long-term operation decisions involve static pricing and
fixing the number of bikes in the system.  Several papers have studied
these issues. See
\cite{katzev2003car,shaheen2007growth,demaio2004will,demaio2009bike,Nair-2}. They
study the static planning problem based on economical aspects and
growth trends.  Another important research direction concerns bike
repositioning: To solve the problem of unbalanced traffic, bikes can
be moved by the operator, either during the night when the traffic is
low (static repositioning) or during the day (dynamic repositioning).
\cite{Raviv-1} study the optimal placement of bikes at the beginning
of a day.  \cite{Chemla-1} develop an algorithm that minimizes the
distance traveled by trucks to achieve a given bike positions,
assuming that bikes do not move (\emph{e.g.}, during the night). The
dynamic aspects are studied by \cite{Contardo-1,Hampshire-1,Raviv-2},
but their model ignores bike moves between periodical updates of the
system. 
These papers rely mainly on optimization techniques.

Redistribution can also be done by users. In most BSS, users have access to the real-time state of the stations,
\emph{e.g.}, by using a smartphone. They can choose to take or to
return their bikes to a station near their destinations.  Moreover
they can be encouraged to do so by the system. For example, Paris, via
the Velib+ system, offers a static reward by giving free time slots in
order to bring more bikes to up-hill stations. It is estimated that
the number of people that obtain rewards via Velib+ during the day is
equivalent to the number of bikes redistributed by trucks. To
compensate for real-time congestion problems, an alternative is to use
real-time pricing mechanisms. This type of congestion control
mechanism is widely applied in the transportation or car-rental
industry, \emph{e.g.},
\cite{guerriero2012revenue,waserhole2012vehicle}. Nevertheless, its
application to BSS is unclear, as the price paid per trip for using
BSS is usually low.

A few papers tackle the stochasticity of BSS, or more generally of
vehicle rental systems (see \cite{Godfrey-1,George-2} and literature
therein). Their first idea was to obtain a simplified asymptotic
behavior when the system gets large. This approximation is valid as
BSS are large systems and can give qualitative and quantitative
properties for the model.  In models with product-form steady-state
distribution (see \cite{George-1,George-2}), the asymptotic expansion
of the partition function can be obtained via complex analysis (saddle
point method), see \cite{Malyshev-6}, or probabilistic tools, see
\cite{Fayolle-7}.  One of the main limitations of these papers is that
they ignore that stations have finite capacities and therefore neglect
the saturation effect.  \smallskip

\textbf{Contributions} -- We present a model of a bike-sharing system
and analyze its steady-state performance. The system is composed of a
large number of stations and a fleet of bikes. Each station can host a
finite number of bikes $K$, called its \emph{capacity}.  We measure
the performance in terms of proportion of so-called {\em problematic
  stations}, \emph{i.e.}, stations with no available bikes or no empty
spots.

We study an homogeneous scenario, in which all stations have the same
parameters. Our aim is to obtain a simple tractable model. Intuitively, this system in
which the flow of bikes between two stations is, on average, identical
in both directions, has roughly speaking the best behavior. For more details, see our work on inhomogeneous
BSS, \cite{velib-aofa}. We investigate the effects of 
random choices of users and characterize the influence of the
station capacity on the performance.
We compare incentives and redistribution mechanisms, and we obtain
closed-form characteristics of their performance. This model can be
straightforwardly extended to model non-homogeneous cities in order to
take into account the difference of attractiveness among the stations.
The paper by \cite{velib-aofa} provides analytical results in the case
where the stations can be grouped into clusters of stations that have
the same characteristics. These results will be briefly presented and
discussed in this paper. The redistribution and incentive mechanisms
in an heterogeneous setting are not studied by \cite{velib-aofa} but
will be included in their upcoming paper,
\cite{velib-2choices-2clusters}.

This homogeneous model also forgets geometry, present in real-world systems 
where the bikes are returned to neighboring stations in case of lack of available spots.
But the homogeneous model study is nevertheless useful: simulations show that the behavior of both systems are very similar.
 Thus the qualitative and quantitative results obtained for the homogeneous model apply to a system with a local search of an empty spot.


Our first contribution is to study the simplest model without any
incentives or redistribution mechanisms. We use a mean-field
approximation that enables us to obtain the asymptotic behavior of our
model as the system size becomes large. 
This asymptotic dynamics leads to simple expressions that give
qualitative and quantitative results. This method works even if a
closed-form (product-form) expression is not available for the
original model.  We show that the proportion of problematic stations
depends on the fleet size and it decreases slowly with the capacity
$K$. For a given capacity, it presents a minimum that is attained at a
fleet size called {\em optimally reliable fleet size}, which is equal
to $K/2+\lambda/\mu$ bikes per station, where $\lambda$ is the arrival
rate of users at a station and $1/\mu$ is the average trip time. This
answers the fleet sizing problem. The term $\lambda/\mu$ represents
the average number of bikes in circulation. It quantifies the 
intuitive idea that the greater the demand is, the more bikes must be
put in the system.  For this fleet size, the proportion of problematic
stations is $2/(K+1)$.

To improve the situation, we investigate two different directions:
incentives and redistribution. We first assume that users have access
to real-time information on the system and follow the rules the system
gives about where to take or return the bike. The improvement that is
obtained in this case is quantified.  We show that returning bikes to
a non-saturated station does not change significantly the behavior of
the stations and the performance with our metric. The situation
improves dramatically when users return their bikes to the least
loaded station among two, even if only a fraction of the users do
this. Indeed we prove that, if all users do this, the proportion of
problematic stations can be as low as $\sqrt{K}2^{-K/2}$. These
results are confirmed by simulations in which users choose among two
neighboring stations. Again, the optimally reliable fleet size is a
little more than $K/2$.

In Section~\ref{sec:regulation}, we then study what we call the
redistribution rate . We define the redistribution rate as the ratio
of the number of bikes that have to be moved manually by trucks over
the number of bikes that are taken by users. It is proved that the
redistribution rate that optimizes performance depends on the fleet
size and the station capacity. More precisely, the redistribution rate
threshold needed to suppress problematic stations is minimal when the
fleet size is $K/2+\lambda/\mu$ bikes per station and is equal to
$1/(K-1)$.

Finally, in Section~\ref{sec:validation}, we discuss the limitations
of the model. We describe briefly the differences that will occur when
considering a time- or space-inhomogeneous model. We mainly refer to
the paper by \cite{velib-aofa}, whose main result is an extension of
the expression of the minimal proportion of problematic stations
within a cluster at some optimally reliable fleet size, which
generalizes $s=K/2+\lambda/\mu$. We also show simulation results of
more realistic models that consider various trip-time distributions
and take the geometry into account.  In all cases, these models
behave similarly to the original model. We also simulate a two-choice
model where the choice is made at the beginning of the trip, when
taking the bike.  We introduce then a different performance indicator,
the average number of stations visited before returning a bike. We show
that, except when the system's geometry is a line, this indicator can
be deduced from the proportion of saturated stations. Moreover we show
that our model can be extended, while remaining analytically
tractable. We detail the case of mean search times that are shorter
than the mean travel time, and the case of losses of users arriving at
an empty station are replaced by a search of an available bike in
another station. These extensions have also been investigated by
simulation on models with geometry. Their behaviors are still very
similar.

\smallskip

\textbf{Organization of the Paper} -- Section~\ref{sec:model} presents
the model description and the mean-field techniques.
Section~\ref{sec:symmetric} deals with the basic model
results. Section~\ref{sec:incentives} studies incentive mechanisms
where a fraction of users choose the least loaded of two stations to
return their bike. Section~\ref{sec:regulation} shows that there
exists a threshold for the rate of redistribution by trucks which
optimizes performance.  Section~\ref{sec:validation} deals with
discussions, extensions and simulation
validations. Section~\ref{sec:conclusion} concludes.

\section{System Model and Mean-Field Analysis}
\label{sec:model}

\renewcommand\chi{\{0, \ldots, K\}}
In this section, we present the basic model.  Mean-field techniques
are used to investigate the performance of homogeneous BSS. These
techniques reduce the study of the stochastic model to the study of
the equilibrium point of a set of differential equations. We detail
the steps to obtain this result.  The other scenarios studied in this
paper fit the same framework.

\subsection{Main Notation List}~

\begin{tabular}{@{}p{1cm}l@{}}
  $N$       & Number of stations.\\
  $s$       & Average number of bikes per station (the total number of bikes is $sN$).\\
  $K$       & Number of slots in a station, also called capacity of the station.\\
  $\lambda$ & Arrival rate of users at a station.\\
  $1/\mu$   & Average trip time.\\
  $Y^N_k(t)$& Proportion of the $N$ stations  where $k$ bikes are available at time $t$.\\
  $y_k(t)$  & Limit of $Y^N_k(t)$ as $N$ tends to infinity (described by an ODE).\\
\end{tabular}

\begin{tabular}{@{}p{1cm}l@{}}
  $\bar{y}$ & Equilibrium point of the corresponding ODE.\\
  $U^N_k(t)$ & Proportion of the $N$ stations where $k$ or more bikes are available at $t$.\\
  $u_k(t)$    & Limit of $U^N_k(t)$ as $N$ tends to infinity (described by an ODE).\\
  $\bar{u}$ & Equilibrium point of the corresponding ODE.\\
\end{tabular}

\subsection{Homogeneous Bike-Sharing Model}\label{sec:hm}

We consider a Markovian model of a bike-sharing system with $N$
stations and a fleet of $\floor{sN}$ bikes ($s$ bikes per station in
average). A bike can be either hosted at a station or in transit
between two stations.  In this paper, we focus on the homogeneous
bike-sharing model. It allows us to obtain a closed-form expression for
the optimal performance and also, in the next sections, to investigate
incentives and redistribution by trucks in this framework, and to
quantify their effects. This study is extended to an
inhomogeneous model by \cite{velib-aofa}.

Each station can host $K$ bikes.  At each station, new users arrive at
rate $\lambda$. If there is no bike at this station, the unhappy user
leaves the system. If the station is not empty, the user takes a bike
at this station and joins the pool of riding users.  The trip time
between the two stations is exponentially distributed with mean
$1/\mu$.

After this time, the riding user wants to return her bike. She chooses
a destination at random among all stations.  If her destination has
fewer than $K$ bikes, the user returns her bike to this station and
leaves the system. If the station has $K$ bikes, no more bikes can be
returned at this station and this station is called
\emph{saturated}. In this case, the user rides to another station.
This station is again chosen at random and the trip time is
exponentially distributed with mean $1/\mu$. This process is repeated
until she finds a non-saturated station.

As seen in the rest of the paper, our model can be thoroughly analyzed
and leads to closed-form results. This model, however, does not
incorporate any geographical information.  In a real-world system,
users who cannot find a bike or cannot return their bike will try a
neighboring station and not just one at random. These modifications
unfortunately lead to intractable models. Nevertheless, we show by
using simulations in Section~\ref{sec:validation} that taking into
account locality has little effect on the overall performance.

\subsection{Mean-Field Limit and Steady-State Behavior}
\label{sec:mf}

In this section, we prove that the analysis of the system is
essentially the analysis of an ordinary differential equation (ODE)
as $N$ goes to infinity and that the equilibrium point can be
addressed.

Let us denote by $Y^N_k(t)$ the proportion of stations where $k$
bikes are available at time $t$.  As the system is homogeneous, the process
$(Y^N(t))=(Y^N_0(t)\dots Y^N_K(t))$ is a Markov process. Suppose the process is at $(y_0\dots y_K)$. There are two types of transitions:
\begin{itemize}
\item \textbf{Bikes taken.} The arrival rate of users in a station
  that has $k$ bikes is $\lambda Ny_k$. When $k\ge1$, it causes the $k$th
  coordinate, $y_k$, to decrease by $1/N$ and $y_{k-1}$ to increase by
  $1/N$.
\item \textbf{Bikes returned.} The number of bikes in transit is equal
  to the total fleet size $sN$ minus the number of bikes hosted at the
  stations, which is equal to $N(s-\sum_{k=1}^Nky_k)$. As trip times
  are exponentially distributed with mean $1/\mu$ and stations are
  chosen at random, a user arrives with a bike at a station with $k$
  bikes at rate $y_k\mu N(s-\sum_{k=1}^K ky_k)$. When $k\le K-1$, this
  causes $y_k$ to decrease by $1/N$ and $y_{k+1}$ to increase by
  $1/N$.
\end{itemize}
The transitions can be summarized by
\begin{align*}
  y & \rightarrow y+\frac{1}{N}(\mathbf{e}_{k-1}-\mathbf{e}_k) & \mathrm{at~rate~}\quad&\lambda y_k
 N\mathbf{1}_{k>0},\\ 
 y & \rightarrow y+\frac{1}{N}(\mathbf{e}_{k+1}-\mathbf{e}_k) & \mathrm{at~rate~}\quad& y_k \mu
 (sN-\sum_{n=0}^K n y_n N)\mathbf{1}_{k<K}
\end{align*}
where the $k$-th unit vector of $\R^{K+1}$ is denoted by
$\mathbf{e}_k$ and $\mathbf{1}_{k<K}$ is equal to 1 when $k<K$ and $0$
otherwise.

This process belongs to the family of \emph{density-dependent
  population processes}, defined by \cite{kurtz1981}. This means that
there exist a set of vector $\mathcal{L}\subset\R^{K+1}$ and a set of
functions $\{\beta_\ell\}_{\ell\in\mathcal{L}}$ such that the
transitions are of the form $y\to y+\ell/N$ and occur at rate
$N\beta_\ell(y)$. This implies that, from $y$, the average change in a
small interval $dt$ is
$f(y)dt=\sum_{\ell\in\mathcal{L}}\ell\beta_\ell(y)dt$.  As $f$ is
Lipschitz-continuous, it is shown in \cite{kurtz1981} that as $N$ goes
to infinity, for each $T>0$, the process $(Y^N(t),0\leq t\leq T)$
converges in distribution to a deterministic function $(y(t),0\leq
t\leq T)$, which is the unique solution of the following differential
system of equations:
 \begin{align}\label{eq:ODE}
   \dot{y}(t) = \sum_{k=0}^K y_k(t) \left( \lambda (\mathbf{e}_{k-1}-\mathbf{e}_k)1_{k>0}+\mu
     (s-\sum_{k=1}^K ky_k(t)) (\mathbf{e}_{k+1}-\mathbf{e}_k)1_{k<K} \right)
\end{align}
where $\mathbf{e}_k$ is the $k$-th unit vector of $\R^{K+1}$.

The first term corresponds to the rate of arrival of new users at a
station, and the second term corresponds to the rate at which users
return bikes, which is $\mu$ times the proportion of bikes in transit
at time $t$.

The above differential equation rewrites $\dot{y}(t)= y(t)L_{y(t)}$
where the $y(t)L_{y(t)}$ is the product of the row vector $y(t)$, by
the jump matrix $L_{y(t)}$. 
This equation contains the mean-field
property of the model: this means that, when $N$ tends to infinity,
the empirical distribution $y(t)$ of the stations evolves in time as
the distribution of some non-homogeneous Markov process on
$\{0,\ldots,K\}$, whose jumps are given by $L_{y(t)}$, updated by the
current distribution $y(t)$. These jumps rates are those of an
$M/M/1/K$ queue, where the arrival rate $\mu (s-\sum_{k=1}^Kky_k(t))$
is time dependent and the service rate is $\lambda$. This queue
represents the instantaneous evolution of any station, because all the
stations have the same evolution due to the homogeneity.

Throughout the paper, we investigate the steady-state behavior of the
system. For all variants of the model studied in this paper, the
dynamical system has a unique equilibrium point.  Note that this fact
alone does not imply that the sequence of invariant measures of $Y^N$
concentrates on this fixed point: it is necessary to show that the
dynamical system does not have long-term oscillations and, in general,
the proof of this is difficult. There are different techniques for
obtaining this result.  In Section~\ref{sec:symmetric}, we show the
absence of oscillations by using a generic Lyapunov function.
Although numerical evidences show that this is also the case for the
other models, the proof is out of the scope of the paper and the
question is not addressed for models with incentives or
redistribution.

\subsection{Performance Metric and Proportion of Problematic Stations}

In this paper, we mainly focus on a quality of service indicator,
called the \emph{proportion of problematic stations}. This proportion
is the proportion of stations where either no bikes are available or
that are saturated.  When the number of stations $N$ goes to infinity,
this proportion converges to $\bar{y}_0+\bar{y}_K$, where $\bar{y}$ is
the unique fixed point of the differential equation.

This metric generalizes the loss probability, used for example in the
context of vehicle rental networks by \cite{George-1}, where the
station capacities are infinite. Moreover, a user will be satisfied if
she can take a bike at her chosen place of departure and return it at
her chosen destination. In an homogeneous system like ours and if
origin and destination are chosen uniformly at random, this occurs
with probability
$1-(1-\bar{y}_0)(1-\bar{y}_K)=\bar{y}_0+\bar{y}_K+\bar{y}_0\bar{y}_K\approx
\bar{y}_0+\bar{y}_K$. Hence, our metric is close to the {\em limiting
  proportion of unsatisfied users}, who cannot enter the system or
return their bike in the station of their choice. It measures the
quality of service in \cite{waserhole2012vehicle} and \cite{Nair-2}.
The average sojourn time in the system can also be deduced from
$\bar{y}_0$ and $\bar{y}_K$.  When a user wants to return a bike at a

saturated station, she has to find a non-saturated station. When this search is
done at random, the average number of stations that are visited before
finding a spot is $1/(1-\bar{y}_K)$ (see details in Section~\ref{sec:other_perf}).  In Section~\ref{sec:geometry},
we show that, for a realistic model of geometry (2D grid), the average
number of visited stations is close to $1/(1-\bar{y}_K)$.

This
justifies the choice of this simple metric, though  others can be
deduced from $\bar{y}_0$ and $\bar{y}_K$.
Indeed, the sum $\bar{y}_0+\bar{y}_K$ hides the relative value of each term,
which can be useful to know.  Nevertheless, this has the advantage of providing a
single indicator of the performance. Our goal in this paper is to
obtain bounds on this performance criteria. The proportion of
problematic stations is not a cost function.  In this paper, we do not
question the cost or the practical methods for implementing our
mechanisms and leave this question for future work. Hence, throughout
the paper, the term \emph{optimal} performance will be understood in
term of minimizing this quality of service. It occurs for a value of
the fleet size called optimally reliable fleet size.


\section{Basic Model and Optimal Fleet Size}
\label{sec:symmetric}
\renewcommand\chi{\{0, \ldots, K\}}

This section is devoted to the basic model. As seen in
Section~\ref{sec:mf}, the behavior of the system can be approximated
by the ODE~\eqref{eq:ODE} when the system becomes large.  This allows
for a complete study of the performance metric as a function of
$s$. We derive the optimal reliable fleet ratio and discuss the
influence of the parameters $\lambda/\mu$ and $K$.

\subsection{Basic  Model: Steady-State Analysis}
\label{sec:basic-model}

An equilibrium point $\bar{y}$ of the ODE~\eqref{eq:ODE} is the stationary
measure of an $M/M/1/K$ queue with arrival rate
$\mu(s-\sum_{k=1}^Kk\bar{y}_k)$ and service rate $\lambda$.  For $\rho\ge0$,
let $\nu_{\rho}$ be the invariant probability measure\footnote{For
  $\rho=1$, $\nu_{\rho}$ is  the uniform distribution on $\chi$. For
  $\rho\ne1$, $\nu_{\rho}$ is geometric: $\nu_{\rho}(k)=\rho^k
  (Z(\rho))^{-1}$ where $Z(\rho)=(1-\rho^{K+1})/(1-\rho)$ is the
  normalizing constant.} of a $M/M/1/K$ queue with arrival-to-service
rate ratio $\rho$ and define $\rho(\bar{y})=\mu(s-\sum_{k\in\chi} k
\bar{y}_k)/\lambda$. The equilibrium points of~\eqref{eq:ODE} are thus the
solutions of the fixed point equation
$\bar{y}=\nu_{\rho(\bar{y})}$.

We now prove the existence and uniqueness of the equilibrium
point. For each $\rho$, there exists a unique stationary measure
$\nu_\rho$. Hence, this equation is equivalent to $\bar{y}=\nu_\rho$
where $\rho$ is the solution of
\begin{align}\label{fixep}
  s=\frac{\lambda}{\mu} \rho + \sum_{k=1}^Kk\nu_{\rho}(k).
\end{align}
This equation can be easily explained. The proportion of bikes per
station $s$ is the sum of two terms: The mean number of users still
riding, $\rho\lambda/\mu$, and the mean number of bikes per station, $
\sum_{k=1}^Kk\nu_{\rho}(k)$. The expression of first term can be
computed using the arrival rate $\lambda$ and the probability of
finding an available bike at a station, $1-\nu_{\rho}(0)$.  Returning
a bike at the $k$-th attempt takes an average time $k/\mu$ and occurs
with probability $(1-\nu_{\rho}(K)) \nu_{\rho}(K)^{k-1}$.  Thus, the
first term is equal to $ \lambda (1-\nu_{\rho}(0))\sum_{k=1}^{+\infty}
(1-\nu_{\rho}(K)) \nu_{\rho}(K)^{k-1} k/\mu $, which reduces, after
some computation, to $\rho \lambda/\mu$.

The right part of Equation~\eqref{fixep} strictly increases in
$\rho$. Therefore, for each $s>0$, there is a unique $\rho$ solution
of equation~\eqref{fixep} and thus a unique equilibrium point
$\nu_{\rho}$ is denoted by $\bar{y}(\rho)$ in the following.

In \cite{velib-aofa}, we exhibit a Lyapunov function for this dynamics
in a more general setting. It shows that all trajectories of the ODE
converge to the unique fixed point. As a consequence, the steady-state
empirical distribution of the system concentrates on this unique fixed
point. This means that the limiting stationary distribution of the
number of bikes at a station is this fixed point, \emph{i.e.}, a
geometric distribution on $\{0,\ldots,K\}$ where the parameter $\rho$
is the solution of Equation~\eqref{fixep}.

\subsection{Proportion of Problematic Stations}
\label{sec:propProblematicStations}

The next theorem shows the effect of the number of bikes per station
$s$ on the performance of the system. Let $\bar{y}$ be the equilibrium
point of the equation.  The fixed point of Equation~\eqref{eq:ODE}
can be rewritten as a polynomial equation in $\rho$ of order
$K+1$. Hence, even if solving this equation is possible for very small
values of $K$, finding a closed-from expression for $K\ge4$ is
unfeasible.  Nevertheless, Equation~\eqref{fixep} provides an
efficient way to achieve a performance study of the system by
considering the parametric curve
\begin{align}
  \lp\rho\frac{\lambda }{\mu} + \sum_{k=1}^K k\rho^k\bar{y}_0(\rho), ~~
  \bar{y}_0(\rho)+\bar{y}_K(\rho)\rp_{\{\rho>0\}}
  \label{eq:parametric}
\end{align}
where, and in the following, $\bar{y}(\rho)$ is equal to $\nu_{\rho}$.
The use of this parametric curve allows us to study efficiently the
performance of the system as a function of the number of bikes per
station $s$. These results are summarized in the next theorem and in
Figure~\ref{fig:performance_onecluster}.

\begin{theorem} 
  \label{th:onecluster}
  \label{th:one-choice}
For the homogeneous model, 
  \begin{itemize}
  \item[(i)] the limiting proportion of problematic stations, $\bar{y}_0+\bar{y}_K$, is minimal
    when $s=K/2+\lambda/\mu$ and the minimum is equal to $2/(K+1)$.  It
    goes to one when $s$ goes to zero or infinity.
    
  \item[(ii)] As $K$ grows, the performance around $s=K/2+\lambda/\mu$
    becomes flatter and insensitive to $s$ and $\lambda/\mu$.
  \end{itemize}
\end{theorem}
\begin{proof}
  Let $\varphi(\rho)=\bar{y}_0(\rho)+\bar{y}_K(\rho) = 1-(\rho^K-\rho)/(\rho^{K+1}-1))$
  and $s(\rho) = \lambda\rho/\mu + \sum_{k=1}^K k \rho^k (1-\rho)/(1-\rho^{K+1}) $. Functions $\varphi$ and
  $s$ are well defined on $[0,\infty)$.  The proportion of problematic stations
  as a function of $s$ is given by $\psi=\varphi \circ s^{-1}$. First we prove
  that $\psi$ has a minimum at $s_0=s(1)$ which is $2/(K+1)$.
  For that, differentiating function $\varphi$ with respect to $\rho$ gives
  \begin{align*}
    \varphi'(\rho)=\frac{\rho^{2K}-1+K(\rho^{K-1}-\rho^{K+1})}{(\rho^{K+1}-1)^2}.
  \end{align*}
  Differentiating the numerator and studying the variation, it holds
  that the numerator strictly increases on $[0,\infty)$ and
  $\varphi'(1)=0$. Thus $\varphi$ strictly decreases on $]0,1[$,
  strictly increases on $]1,+\infty[$ and has a minimum at $1$ which
  is $\varphi(1)=2/(K+1)$. This shows the optimal number of bikes per
  station corresponds to $\rho=1$, and thus $s(1)=K/2+\lambda/\mu$ is
  the optimal reliable proportion of bikes per station. This leads to
  a proportion of problematic stations of $2/(K+1)$ and concludes the
  proof of \emph{(i)}.

  To prove \emph{(ii)}, note that the second derivative of $\psi=\varphi
  \circ s^{-1}$ at $s_0=s(1)$ 
  is $\psi''(s_0)=\varphi'' (1)/s'(1)^2$. An asymptotic expansion
  of $\varphi$ at $\rho=1$ is given by
  \[
  \varphi(\rho)=\frac{2}{K+1}+\frac{1}{6} \frac{K(K-1)}{K+1}(\rho-1)^2+O((\rho-1)^3)
  \]
  and therefore  $\varphi'' (1)=K(K-1)/3(K+1).$ Moreover,
  differentiating $s$ gives that $s'(1)=\lambda/\mu+K^2/12+K/6$ which
  leads to
  \begin{align*}
    \psi'' (s_0)=\frac{K(K-1)}{3(K+1)(\lambda/\mu+K^2/12+K/6)^2} \sim_{K\to\infty}\frac{48}{K^3}.
  \end{align*}
  This means that $\psi$ is never sharp for the range of values
  considered in this paper, \emph{i.e.}, $K\geq 10$. Moreover, $ \psi''
  (s_0)$ goes quickly to $0$ as $K$ grows.
   \qed
\end{proof}

This theorem indicates that, even for a homogeneous system for which
the number of bikes per station is chosen knowing all parameters of
the users, the proportion of problematic stations decreases only at
rate $1/K$. This is problematic for practical situations where, for
space constraints and construction costs, station capacities are often
fewer than 20 or 30 bikes. A system with $30$ bikes per station would
lead to a proportion of problematic stations of $2/31\approx
6.5\%$. Although it might be acceptable if this bike-sharing system is
used only once in a while, a probability of $6.5\%$ for a station to
be problematic is too large for a reliable daily mode of
transportation.

\begin{figure}[ht]
  \centering
  \begin{tabular}{cc}
    \subfigure[\label{fig:K30}The capacity of the stations is $K=30$.]
    {\includegraphics[width=.47\linewidth]{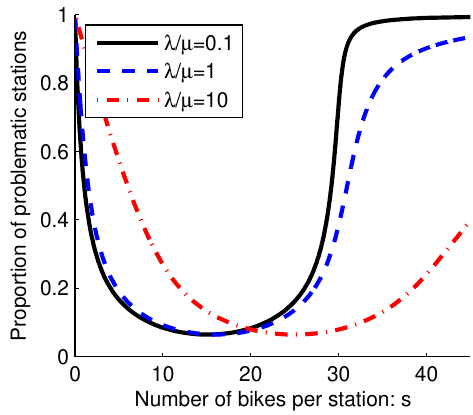}}&
    \subfigure[\label{fig:K100}The capacity of the stations is $K=100$.]
    {\includegraphics[width=.47\linewidth]{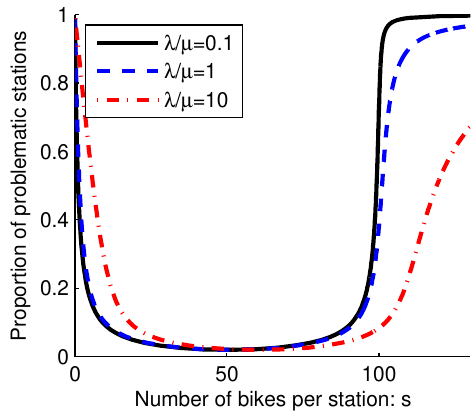}}
  \end{tabular}
  \caption{Proportion of problematic stations as a function of the
    number of bikes per station for two values of the size of stations
    $K$.  On the $x$-axis is the average number of bike per station
    $s$. For both scenarios, we plot $\lambda/\mu=0.1$, $\lambda/\mu=1$ and
    $\lambda/\mu=10$.}
  \label{fig:performance_onecluster}
\end{figure}

The results of Theorem \ref{th:onecluster} are illustrated by
Figure~\ref{fig:performance_onecluster}, which plots the performance  as the parametric curve given by Equation~\eqref{eq:parametric} for two values of the station capacity and three values of 
$\lambda/\mu$.

When $K$ is fixed to $30$ and $\lambda/\mu=1$, the performance is
almost equally good with $10$ to $20$ bikes per station (\emph{i.e.},
$s\in[K/3;2K/3]$).  However, as soon as the number of bikes per
station is lower or higher, the performance decreases
significantly. The problematic case is mainly due to empty stations at
low $s$ versus saturated ones at large $s$. When the size of the stations
is $K=100$, the performance is less sensitive to the number of bikes
per station. As pointed out by Theorem~\ref{th:onecluster}, in this
case, the proportion of problematic stations is $2/(K+1)\approx 2\%$:
Multiplying the station capacity by $3$ divides by $3$ the minimum
proportion of problematic stations. Having some stations designed to
host up to $100$ bikes is realistic for stations near a subway for
example, but having all stations in a city with $100$ slots is very
costly in terms of space and installation.

When $\lambda/\mu=10$, the situation is similar to the case
$\lambda/\mu=1$, with curves shifted to the right. The minimum number
of problematic stations is the same, only the optimal reliable fleet
size is changed.  This can be deduced from \eqref{eq:parametric} as
the term $\lambda/\mu$ affects only $s(\rho){=}\rho \lambda/\mu {+}
\sum k\rho^k \bar{y}_0(\rho)$ and not the proportion of problematic
stations $\bar{y}_0(\rho){+}\bar{y}_K(\rho)$.

These results suggest that without incentives for users to return
their bikes to a non-saturated station or without any load-balancing
mechanisms, the implementation of a bike-sharing system will always
observe a poor performance, even if the system is homogeneous
and there are no preferred areas. In a real system where some regions
are more crowded than others (\emph{e.g.}, because of the trips from
residential areas to work areas), the situation can only be worse. See
\cite{velib-aofa} for a study.  In the following, we will examine
simple mechanisms that improve dramatically the situation.

\subsection{If People  Return the Bikes to a Non-saturated Station}
\label{sec:avoiding}

Before studying incentive or regulation mechanisms, in this section we
study a variant of the model where users know which stations are empty
or saturated. They always arrive to non-empty stations and return
their bikes only to non-saturated stations.

The dynamics of the system are slightly modified as follows.  As
before, there is a Poisson arrival process in the system at rate
$N\lambda$, but each arriving user picks at random a station among the
non-empty stations. If there are no non-empty stations, she leaves the
system. If the user manages to find a bike, then after a time
exponentially distributed with parameter $\mu$, she arrives at a
non-saturated station picked at random (\emph{i.e.}, with fewer than
$K$ bikes), returns her bike at this station and leaves the system.
Note that there is always a non-saturated station (this is the case if
$s<K$).

The transitions of $(Y^N(t))$ are now given by
\begin{align*}
  y & \rightarrow y+\frac{1}{N}(e_{k-1}-e_k) & \frac{\lambda y_k}{1-y_0} N1_{k>0,\; y_0<1}\\
  y & \rightarrow y+\frac{1}{N}(e_{k+1}-e_k) & \frac{y_k}{1-y_K} \mu \lp
  sN-\sum_{n=0}^K n y_n N\rp1_{k<K,\; y_K<1}.
\end{align*}
The differential equation is replaced by $\dot{y}=f(y)$ where
\begin{equation*}
  f(y) = \sum_{k=0}^K y_k\left( \frac{\lambda}{1-y_0}
    (e_{k-1}-e_k)1_{k>0}+\frac{\mu
      (s-\sum_{k=1}^Kky_k)}{1-y_K}(e_{k+1}-e_k)1_{k<K}\right).
\end{equation*}
The function $f$ is discontinuous when $y_0$ goes to $1$ or when $y_K$
goes to $1$. Nevertheless, it can be shown with elementary arguments
that if $\lambda/\mu<s<K+\lambda/\mu$ then the differential equation
$\dot{y}=f(y)$ has a unique solution.  In the following, we assume
that $\lambda/\mu<s<K+\lambda/\mu$. As for the previous model,
equation $\dot{y}=f(y)$ can be rewritten as $\dot{y}=yL_y$. This time,
$L_y$ is the infinitesimal generator of an $M/M/1/K$ queue with
arrival rate $\mu(s-\sum_k ky_k)/(1-y_K)$ and service rate
$\lambda/(1-y_0)$.

By the same method as in the previous section, it can be proved that
the dynamical system has a unique fixed point $\bar{y}$, which is solution
of the equation
$s=\lambda/\mu + \sum_{k=1}^Kk\bar{y}_k(\rho)$.
Moreover, following the same lines as Tibi \cite[Proposition
4.3]{Tibi-1}, the steady-state, denoted by $Y^N(\infty)$, converges as
$N$ gets large, to this fixed point.  The limiting steady-state of the
number of bikes at a station is, again, geometrically distributed,
with parameter $\rho$ given by the previous equation.  The fact that
the term $\rho \lambda/\mu$ is replaced here by $\lambda/\mu$ can be
simply explained. Each user is accepted in the system and returns the
bike after one trip with mean time $1/\mu$. There are, on average,
$\lambda/\mu$ users riding per station.

When studying the fixed point of the system, we find that the main
difference with the original model studied in Section
\ref{sec:propProblematicStations} is the expression of $s$. Therefore,
the performance of the system is easily plotted by a parametric curve
similar to \eqref{eq:parametric}.
\begin{figure}[ht]
  \centering
  \begin{tabular}{cc}
    \includegraphics[width=.7\linewidth]{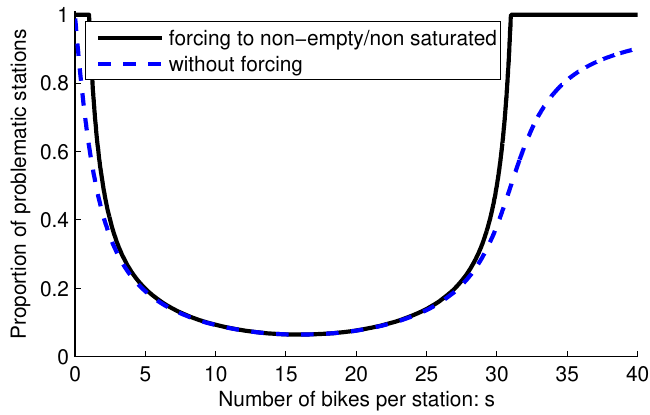}
  \end{tabular}
  \caption{Proportion of problematic stations as a function of the number
    of bikes per station when we force people to go to a non-saturated
    station compared to the proportion if we do not force people. Values
    for $K=30$ and $\lambda/\mu=1$.}
  \label{fig:go_full}
\end{figure}
The proportion of problematic stations has a similar shape for this
model and the one of the previous section. As before, the minimal
proportion is $2/(K+1)$ and is attained for $s=K/2+\lambda/\mu$. When
$s$ is not equal to $K/2+\lambda/\mu$, this proportion will be higher
than in the classic model.  This is illustrated in
Figure~\ref{fig:go_full} where the proportion of problematic stations
for that model is compared with the model studied in Section
\ref{sec:basic-model}.  For small $s$, the less satisfactory behavior
of the regulated system is due to the fact that every user can enter
the system. This resembles the influence of a larger $\lambda$ in
Figure 1. For large $s$, this poor behavior of the system is related
to the improvement on the customer trip time. It has the same effect
as the influence of $1/\mu$ in the basic system.

This shows that although forcing people to go to a non-saturated or
non-empty station reduces the unhappy users because anyone can take or
leave a bike at anytime, it makes the system more congested and
degredate the situation for users who are not aware of such
mechanisms.


\section{Incentives and the Power of Two Choices} 
\label{sec:incentives}

In this section, we consider that, when a user wants to return her
bike somewhere, she indicates two stations and the bike-sharing system
indicates to her which one of the two has the least number of bikes
available. We show that, when the two stations are picked \emph{at
  random}, the proportion of problematic stations diminishes as
$\sqrt{K}2^{-K/2}$ (instead of $1/K$ in the original model). The
performance is thus improved dramatically, and even if only a small
percentage of users obey this rule. This result is similar to the
well-known \emph{power of two choices} that has been proved to be a
very efficient load balancing strategy, see \cite{mitzenmacher}.

\subsection{The Two-Choice Model and Its Steady-State Analysis}

We consider an homogeneous model with $N$ stations and $s$ bikes per
station. As before, users arrive at rate $\lambda$ in each station and
take a bike if the station is not empty. Otherwise, they leave the
system. When a user chooses her destination, instead of choosing one
station, she picks two stations at random, travels and returns the
bike to the one that has the lowest number of bikes available.

Let $u_k(t)$ be the proportion of stations with $k$ or more bikes at
time $t$ ($k\in\{0\dots K\}$) \emph{i.e.},
$u_k(t)=y_k(t)+\dots+y_K(t)$. The state of the system can be described
by the vector $(u_k(t))_{k\in \{0,\ldots K\}}$ that is such that
$u_K(t)\le u_{K-1}(t)\le \dots \le u_0(t)=1$.

There are two types of transitions for the Markov process. Suppose $(u(t))$ is at $u=(u_0,\ldots,u_K)$. The first
one is a transition from $u$ to $u-e_k$, when a user that takes a bike from a station with $k$ bikes. This
happens at rate $N \lambda(u_{k}-u_{k+1}1_{\{k<K\}})$.  The second type is a transition from $u$ to $u+e_k$, when a
user returns a bike. The number of bikes locked at  stations is $\sum_{k=0}^K
k y_k=\sum_{k=1}^K u_k$. Hence, there are $N(s-\sum_{k\geq 1} u_k)$
bikes in transit. As a user chooses the least loaded among two
stations, a user returns a bike at a station with $k-1$ bikes at rate $\mu N
(u_{k-1}^2-u_k^2)(s-\sum_{k=1}^K u_k)$. As in Section~\ref{sec:model},
as $N$ grows large, the behavior of the system can be approximated by
the dynamics of the following ODE:
\begin{align}
  \dot{u_k}(t) &= -\lambda (u_{k}(t)-u_{k+1}(t)) + \mu (u_{k-1}^2(t)
  -u_k^2(t))(s-\sum_{k=1}^K u_k(t)),
  \label{eq:power2}
\end{align}
for $k\in\{1,\ldots,K\}$ and $u_0(t) =1$ and $u_\ell(t)=0$ for
$\ell>K$.

The following theorem shows that these incentives dramatically improve
the performance compared to the original model where users go to a
station at random (Theorem~\ref{th:onecluster}). For a given capacity
$K$, the optimal proportion of problematic stations goes from $2/(K+1)$ in the
original model to $\sqrt{K}2^{-K/2}$. 
\begin{theorem}
  \label{th:two-choice}
  Assume that all users obey to the two-choice rule. Then, the
  corresponding dynamical system, given by~\eqref{eq:power2} has a
  unique fixed point.  
  The proportion of problematic stations is lower than
  $4\sqrt{K}\;2^{-K/2}$ for all $s\in[K/2+\lambda/\mu;K-\log_2K-3+\lambda/\mu]$.
\end{theorem}


\begin{proof}
  First show that the ODE~\eqref{eq:power2} has a unique fixed
  point.  Let $\rho:=\mu(s-\sum_{k\geq 1} \bar{u}_k)/\lambda$. The key point is to reduce the equation giving the fixed point to a first order recurrence equation. Indeed, a direct recurrence
  gives that a fixed point must satisfy $\bar{u}_0=1$ and for all $k\le K$:
  \begin{align}\label{JSQ2}
    \bar{u}_{k+1}=\rho(\bar{u}_k^2-1)+\bar{u}_1
  \end{align}
  with $\bar{u}_{K+1}=0$.
  
  For all $k\ge1$ and $x\in[0;1]$, define $\bar{u}_k(x)$ by $\bar{u}_0(x)=1$
  and $\bar{u}_{k+1}(x) = \rho(\bar{u}_k(x)^2-1)+x$.
  By induction on $k$, there is a increasing sequence
  $x_1<x_2<x_3\dots<\rho$ such that $x\mapsto \bar{u}_k(x)$ is strictly
  increasing on $[x_{k-1},\rho]$ and $\bar{u}_k(x_k)=0$. Indeed,
  \begin{itemize}
  \item This is true for $k=1$ because $\bar{u}_1(x)=x$. 
  \item Then, if it is true for some $k\ge1$, then $\bar{u}_{k+1}$ is
    increasing on $[x_k,\rho]$ because $x\mapsto \bar{u}_{k}(x)$ is increasing
    and positive on $[x_k;\rho]$. Moreover, 
    $\bar{u}_{k}(x_{k-1})=x_{k-1}-\rho<0$ and $\bar{u}_{k+1}(\rho)=\rho(\bar{u}_{k}^2(\rho)-1)+\rho=\rho(\bar{u}_k^2(\rho))>0$.
  \end{itemize}
  This shows that there exists a unique $x_k$ such that $\bar{u}_k(x)=0$ on
  $[x_{k-1},1]$. 
  
  A fixed point of Equation~\eqref{eq:power2} is a vector
  $(\bar{u}_1(x),\bar{u}_2(x),\dots \bar{u}_K(x))$ such that $\bar{u}_k(x)\ge0$ and
  $\bar{u}_{K+1}(x)=0$. By the property stated above, for a fixed $\rho$,
  there is a unique fixed point, which is 
  $(\bar{u}_1(x_{K+1}),\dots \bar{u}_K(x_{K+1}))$.  Moreover,
  $s=\sum_{k\geq 1}\bar{u}_k+\rho\lambda/\mu$ is increasing in $\rho$, which implies
  that there is a unique fixed point $\rho$ when $s$ is fixed.
  
  Let $k\le K$ and assume that $\rho\le1$.  If $\bar{u}_1\ge\rho$, a direct
  recurrence shows that $\bar{u}_{k}\ge \rho^{2^{k}-1}$ for $k\le K+1$ ,
  which contradicts the fact that $\bar{u}_{K+1}=0$. Therefore, 
  $\bar{u}_1<\rho$. Hence, for all $1\le k\le K$, 
\begin{equation*}
  \bar{u}_{k+1}=\rho (\bar{u}_k^2-1)+\bar{u}_1= \rho \bar{u}_k^2+(\bar{u}_1-\rho)< \rho \bar{u}_k^2 \le \bar{u}_k^2
\end{equation*}
which implies that $\bar{u}_{k}\le \bar{u}_1^{2^{k-1}}$. As $\bar{u}_{K+1}=0$,
using Equation~\eqref{JSQ2},
\begin{equation}
  0 = \bar{u}_{K+1} = \rho (\bar{u}_K^2-1)+\bar{u}_1 \le  \rho(\bar{u}_1^{2^{K}}-1)+ \bar{u}_1\le \bar{u}_1^{2^{K}}+ \bar{u}_1-\rho.
  \label{eq:2choice1}
\end{equation}

Let $\delta=1-\rho$ and $\varepsilon=\rho-\bar{u}_1\ge0$.
Using Equation~\eqref{eq:2choice1},
\begin{equation}
  0\le(\rho-\varepsilon)^{2^K}-\varepsilon  \le (1-\varepsilon)^{2^K}-\varepsilon\le \exp(-2^K\varepsilon)-\varepsilon.
  \label{eq:2choice2}
\end{equation}
If $\varepsilon\ge K2^{-K}$, then
$(1-\varepsilon)^{2^K}-\varepsilon\le\exp(-K)-K2^{-K}$,
which is less than $0$ for all $K\ge1$. This contradicts
Equation~\eqref{eq:2choice2} and shows that
$\varepsilon\le K2^{-K}$.

The proportion of empty stations is $\bar{y}_0=1-\bar{u}_1=\delta+\varepsilon$. The
proportion of saturated stations is $\bar{y}_K=\bar{u}_K$, which is such that
$\rho(\bar{u}_K^2-1)+\bar{u}_1=0$. Thus, $\bar{u}_K=\sqrt{(\rho-\bar{u}_1)/\rho}$. This shows
that, for all $\rho\in[1-2^{-K/2};1]$, the proportion of problematic
stations is less than
\begin{equation*}
  \bar{y}_0+\bar{y}_K = \delta+\varepsilon + \sqrt{\varepsilon/\rho}\le2^{-K/2}+K2^{-K}+\sqrt{\frac{K2^{-K}}{1-2^{-K/2}}}.
\end{equation*}
This quantity is less than $4\sqrt{K}2^{-K/2}$ for all $K\ge1$ and is
asymptotically equivalent to $\sqrt{K}2^{-K/2}$.

The fleet size $s$, equal to $\sum_{k=1}^K\bar{u}_k+\rho\lambda/\mu$, is an increasing function of $\rho$. Moreover,  $\bar{u}_k\le
\bar{u}_1^{2^{k-1}}\le\rho^{2^{k-1}}.$ Hence, when $\rho=1-2^{-K/2}$, 
\begin{equation*}
  \sum_{k=1}^K\bar{u}_k\le \sum_{k=1}^K(1-2^{-K/2})^{2^{k-1}}\le \sum_{k=1}^K
  \exp(-2^{k-1-K/2})=\sum_{i=-K/2}^{K/2-1}\exp(-2^{i})< K/2.
\end{equation*}
This shows that if $s\ge K/2+\lambda/\mu$, 
$\rho\ge1-2^{-K/2}$.

When $\rho=1$, a direct induction on $k$ shows that
$\bar{u}_k\ge\max(0,1-2^k\varepsilon)$. Let $j$ be such that
$1-2^j\varepsilon\ge0>1-2^{j+1}\varepsilon$.  For such a $j$,
$j+1\ge\log_2\varepsilon=K-\log_2K$. Hence
\begin{align*}
  \sum_{k=1}^K\bar{u}_k \ge \sum_{k=1}^j 1-2^k\varepsilon \ge j-(2^{j+1}-2)\varepsilon\ge j-2^{j+1}\varepsilon\ge K-\log_2K-3.
\end{align*}
This implies that, for all
$s\in[K/2+\lambda/\mu;K-\log_2K-3+\lambda/\mu]$, 
$\rho\in[1-2^{-K/2};1]$.\qed
\end{proof}

Assume now that only a fraction of the users follow this rule and that
the others go to a station at random. This could happen if the users
are rewarded, when they obey the two-choice rule, like in the Velib+
system. To model this behavior, we assume that each user obeys to the
\emph{two-choice rule} with probability $r$ and otherwise chooses only
one station and returns the bike to it.  The dynamics are similar to
Equation~\eqref{eq:power2} and an equilibrium point $\bar{u}$
satisfies the following equations: $\bar{u}_0=1$ and
\begin{align}
  \bar{u}_k-\bar{u}_{k+1}=\rho \lp r(\bar{u}_{k-1}^2
  -\bar{u}_k^2)+(1-r)(\bar{u}_{k-1}-\bar{u}_k)\rp\text{ for }k\in\{1\ldots K\}.
  \label{eq:two-choice_r}
\end{align}
In this case, the fixed point Equation~\eqref{JSQ2} becomes
$\bar{u}_{k+1}=\rho(r(1-\bar{u}_k^2)+(1-r)(1-\bar{u}_k))+\bar{u}_1$.  The proof of the
uniqueness of the solution of Equation~\eqref{JSQ2} can be easily adapted to
show that Equation~\eqref{eq:two-choice_r} also has a unique fixed
point.


\subsection{Impact on the Performance}
\label{sec:performance}
 Due to the uniqueness of the solution $\bar{u}(\rho)$, 
the proof of Theorem~\ref{th:two-choice}, especially Equation~\eqref{JSQ2}, provides an efficient way to
compute $u$ as a function of $\rho$. This shows that, if $\rho$ is
fixed, the number of bikes in the system is $s(\rho)=\lambda\rho/\mu +
\sum_{k=1}^K \bar{u}_k(\rho)$. The performance indicator can be plotted  by
using a parametric curve of parameter $\rho$.  These results are
reported in Figure~\ref{fig:power2} and indicate that the performance
of the system is radically improved compared to the original case
(Figure~\ref{fig:performance_onecluster}), even if $20\%$ of users
obey the two-choice rule.

\begin{figure}[ht]
  \centering
  \begin{tabular}{cc}
    \subfigure[\label{fig:r=1}Proportion of problematic stations as a
    function of the number of bikes per station $s$ for
    $r=1$ (everyone obeys the rule) and $r=0$.]{\includegraphics[width=.47\linewidth]{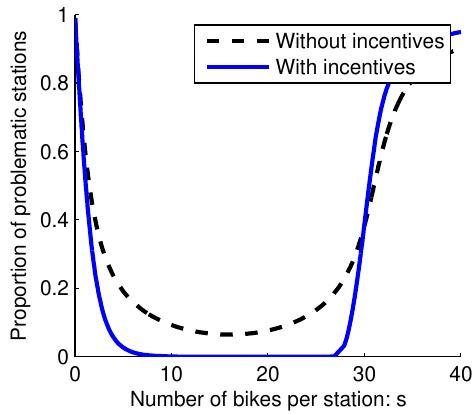}} &
    \subfigure[\label{fig:r_varies}Proportion of problematic stations as a
    function of the 
    proportion of users obeying the two-choice rule (the $y$-axis is in logscale).] 
    {\includegraphics[width=.47\linewidth]{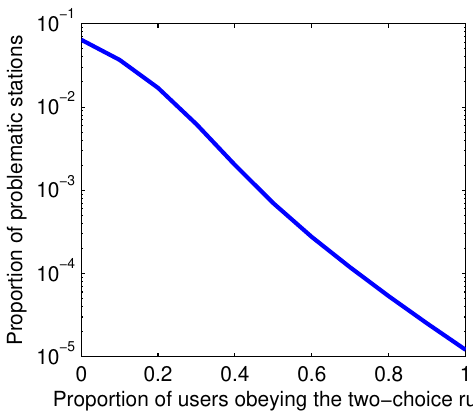}} 
  \end{tabular}
  \caption{Performance of the system in a two-choice system: \ref{fig:r=1}
    when everyone obeys the two-choice rule and the number of bikes per
    station varies. \ref{fig:r_varies}: when the number of bikes per
    station is $16$ and $r$ varies ( $K=30$ and
$\lambda/\mu=1$).}
\label{fig:power2}
\end{figure}

In Figure~\ref{fig:r=1} we report the proportion of problematic
stations as a function of the proportion of bikes per station when
everyone follows the two-choice rule. We observe that the optimal
performance of the system is much better than in the original system
(here $K=30$ and $\lambda/\mu=1$). Although in the original system,
the proportion of problematic stations is at best around $7\%$, here
the proportion of problematic stations can be as low as $10^{-6}$
(this is lower than the bound of Theorem~\ref{th:two-choice}, which is
$2\sqrt{31}\cdot2^{-30/2}\approx 3\cdot 10^{-4}$). Moreover, this
curve is rather insensitive to variations in the number of bikes: The
proportion of problematic stations is less than $10^{-3}$ if $s$ is
between $10$ and $27$ bikes. An interesting phenomenon occurs when the
average number of bikes per station $s$ exceeds the capacity of the
station. In this case, there is a larger proportion of problematic
stations for the two-choice model than for the original
situation. This is explained by the fact that when a user obeys the
two-choice rule, it is easier for her to return a bike. Hence, when
more users obey the two-choice rule, there are fewer bikes in transit
and the stations are more occupied. This negative effect only occurs
when $s\ge30$, which confirms Theorem~\ref{th:two-choice}: Performance
is low for $s$ less than a value $K-\log_2K-3+\lambda/\mu$ close to
$K$.

On Figure~\ref{fig:r_varies}, the average number of bikes per station
is fixed, $s=16=K/2+\lambda/\mu$, and the proportion $r$ of users who
obey the two-choice rule varies from $0$ to $1$. This shows that the
proportion of problematic stations diminishes rapidly as soon as the
number of users obeying the rule grows. Moreover, the decrease is
approximately exponential: if $25\%$ more users obey the rule, the
proportion of problematic stations is roughly divided by $10$.


\section{Optimal Redistribution Rate}
\label{sec:regulation}
Balancing the number of bikes in various areas in a city is one of the
major issues of bike-sharing systems. A widely adopted solution is the
use of trucks to move bikes from saturated stations to empty
ones. This redistribution mechanism can equalize the one-directional
flows of travelers, for example from residential areas to work areas,
but also the imbalances due to the choices of users.  In this section
the minimal redistribution rate needed to suppress any problematic
station is investigated, and we conclude by showing that it decreases
as the inverse of the station capacity.  Our analysis assumes that
bikes are moved one by one. This assumption will be relaxed in
simulations in Section~\ref{sec:truck_capacity}, thus showing that a
larger truck size does not affect qualitatively the performance.

\subsection{The Redistribution Model and Its Steady-State Analysis}

We consider a homogeneous model of bike-sharing systems with $N$
stations equipped with a truck that visits the stations to adjust
their load. The user behavior is the same as in
Section~\ref{sec:model}. The arrival rate at any station is $\lambda$
and a bike trip takes a time exponentially distributed of mean
$1/\mu$. A truck knows the station occupancies at any time. It goes to
the most loaded station, takes a bike and returns it to the least
loaded station. The trip time of the truck is neglected (the bikes are
assumed to move instantaneously from highly loaded to lightly loaded
stations). This description amounts to one truck that moves bikes at
rate $N\gamma$. The resulting Markov model is the same if there is
$\delta(N)$ trucks moving bikes at rate $N\gamma/\delta(N)$.  In the
rest of the section, we study the effect on the performance of the
ratio $\gamma/\lambda$. This ratio is the average number of bikes per
second that are moved manually by the operator divided by the number
of bikes per second taken by regular users.

As in Section \ref{sec:incentives}, let $u_k(t)$ be the proportion of
stations that have $k$ bikes or more available at time $t$.  In
particular, $u_0(t)=1$ and $u_{K+1}(t)=0$. There are three kinds of
transitions in the system: arrivals and departures of users and
redistribution. The fluid model transitions corresponding to the user
arrivals and departures are the same as in Equation~\eqref{eq:ODE}
with a number of bikes in transit of $N(s-\sum_{k=1}^Ku_k)$.
Moreover, the redistribution part only affects the most loaded and
least loaded stations, \emph{i.e.}, stations that have $k$ bikes
available, where $k$ is such that no stations have less less than $k$
bikes available ($u_{k-1}{=}1$) or no stations have more than $k$
bikes available ($u_{k+1}=0$). This shows that the expected variation
of $u$ during a small time interval is equal to $f(u)=(f_0(u)\dots
f_K(u))$, where for all $k\in\{1,\dots,K{-}1\}$,
\begin{equation*}
  f_k(u)=\lambda(u_{k+1}-u_k) + \mu\left(s-\sum_{k=1}^Ku_k\right) (u_{k-1}-u_k) %
  + \left\{
    \begin{array}{ll}
      \gamma & \mathrm{if~}u_{k-1}=1\mathrm{~and~}u_k<1\\
      -\gamma & \mathrm{if~}u_{k+1}=0\mathrm{~and~u_k>0}\\
      0 & \mathrm{otherwise}
    \end{array}
  \right.
\end{equation*}
The function $f(u)$ is not continuous in $u$. Hence, the ODE
$\dot{u}=f(u)$ is not well defined and can have zero solutions. To
overcome this discontinuity problem, it has been shown by
\cite{gast-mama,gast-DI}, that this ODE can be replaced by a
differential inclusion $\dot{u}_k\in F(u)$, where $F(u)$ is the convex
closure of the set of values $f(u')$ for $u'$ in a neighborhood of
$u$. This differential inclusion is a good approximation of the
stochastic system as $N$ grows. In particular, as with classic ODE, if
all the solutions of the differential inclusion converge to a fixed
point, then their stationary measures concentrate on this point as $N$
goes to infinity, see \cite{gast-DI}.

In our present case, the differential inclusion is
\begin{align}
  \dot{u}_k &\in \lambda(u_{k+1}-u_k) +
  \mu\left(s-\sum_{k=1}^Ku_k\right) (u_{k-1}-u_k) %
  + G_k(u),\label{eq:DI}\\
  &\quad\text{where } G(u) = \left\{ (a_0\dots a_i, 0\dots0,-b_j\dots
    -b_K) \text{ s.t. } \left\{
      \begin{array}{ll}
        \sum_{k=1}^i a_k = \sum_{k=j}^K b_k = \gamma\\
        a_k =0 \text{ if } u_k<1\\
        b_k =0 \text{ if } u_{k+1}>0\\
        a_k\ge0, b_k\ge0.
      \end{array}
    \right.
  \right \}\nonumber
\end{align}
To ease the presentation, the details of its construction are
omitted. The construction is similar to Section~4.3 of \cite{gast-DI}.
This leads to the following result.
\begin{theorem}
  \label{th:DI}
  Assume that a truck moves bikes from the most loaded station to
  the least loaded one at rate $\gamma N$. If $x=\min(s{-}\lambda/\mu,
  K{-}s{+}\lambda/\mu)$ and
  $\gamma^*=2\lambda\frac{\floor{2x}-x}{\floor{2x}\floor{2x-1}}$ then
  the fixed point of the dynamical system~\eqref{eq:DI} satisfies:
  \begin{itemize}
  \item If $\gamma\ge\gamma^*$, then there is no problematic station.
  \item If $\gamma<\gamma^*$, the proportion of problematic stations
    decreases with $\gamma$.
  \end{itemize}
\end{theorem}
The quantity $\gamma^*$ is called the \emph{optimal redistribution
  rate}. Setting $\gamma=\gamma^*$ is not necessarily optimal in terms
of cost but it corresponds to a key of the performance: When
$\gamma<\gamma^*$, the proportion of problematic station decreases
almost linearly and is zero when $\gamma>\gamma^*$.

\begin{proof}
  Let us assume that $s\le K/2+\lambda/\mu$. The other case is
  symmetric and can be treated similarly. The proportion of
  problematic stations is non-increasing in the redistribution rate
  $\gamma$. We define the vector $u=(u_0,u_1,\dots,u_K)$ by
  \begin{equation*}
    u_k=\left\{
      \begin{array}{ll}
        1-(k-1)\frac{\gamma}{\lambda}&\text{ if } k\le\floor{2x}=\floor{2(s-\frac{\lambda}{\mu})}\\
        0 & \text{ otherwise}
      \end{array}
    \right.
  \end{equation*}
  Let us show that when $\gamma=\gamma^*$, $u$ is a fixed point of the
  differential inclusion~\eqref{eq:DI}, \emph{i.e.}, there exists $g\in
  G(u)$ such that Equation~\eqref{eq:DI} is equal to zero.  Let
  $g=(g_0,g_1,\dots,g_K)$ be a vector such that $g_1=\gamma$,
  $g_{\floor{2s}}=\lambda u_{\floor{2s}}-\gamma$,
  $g_{\floor{2s}+1}=-\lambda u_{\floor{2s}}$ and $g_i=0$
  otherwise. The vector $g$ belongs to $G(u)$, defined in
  Equation~\eqref{eq:DI}.
  
  Moreover, by a direct computation,
  \begin{align*}
    \sum_{k=1}^K u_k = \sum_{k=1}^{\floor{2x}} \lp%
    1-2(k-1)\frac{\floor{2x}-x}{\floor{2x}\floor{2x-1}}\rp =
    x=s-\frac{\lambda}{\mu}.
  \end{align*}
  In particular,  $\mu(s-\sum_{k=1}^Ku_k)=\lambda$. 
  Plugging it in Equation~\eqref{eq:DI},
  \begin{itemize}
  \item for $k=1$,  $\lambda(u_{2}-u_1)
    \mu(s-\sum_{k=1}^Ku_k)\lambda (u_{0}-u_1) +
    \gamma=\lambda(u_2-u_1)+\gamma=0$.
  \item for $k\in\{2\dots\floor{2s}{-}1\}$, $\lambda(u_{k+1}-u_k) +
    \mu(s-\sum_{k=1}^Ku_k)(u_{k-1}-u_k)=\lambda \gamma-\lambda\gamma=0$.
  \item for $k=\floor{2s}$, $\lambda(u_{\floor{2s}+1}-u_{\floor{2s}}) +
    \mu(s-\sum_{k=1}^Ku_k)(u_{\floor{2s}-1}-u_{\floor{2s}})+\lambda u_{\floor{2s}}-\gamma=0$
  \item for $k=\floor{2s}+1$,
    $\lambda(u_{\floor{2s}+2}{-}u_{\floor{2s}+1}) +
    \mu(s{-}\sum_{k=1}^Ku_k)(u_{\floor{2s}}-u_{\floor{2s}+1})+\lambda
    -\gamma u_{\floor{2s}}=0$
  \end{itemize}
  This proves that $u$ is a fixed point of the differential inclusion
  \eqref{eq:DI}. Using monotonicity arguments  as in
  Theorem~\ref{th:two-choice}, we can show that this fixed point is
  unique. However, the proof is quite technical, hence omitted. As the
  proportion of problematic stations is zero for $u$, this concludes
  the proof of the theorem.

  We now consider the case $\rho<1$ (which corresponds to
  $s<K/2+\lambda/\mu$ and $\gamma<\gamma^*$). Define
  $z(\rho):=(1-\rho+(\rho^{K}-\rho)\gamma)/(1-\rho^{K+1})$ and a
  sequence $x(\rho)=(x_0\dots x_K)$:
  \begin{itemize}
  \item If $\rho^{K-1}(\rho z+\gamma) > \gamma$, then  define
    $x_k(\rho)=\rho^kz(\rho)$ for $k\in\{1\dots K-1\}$ and
    $x_K(\rho)=\rho^K z(\rho)-\gamma$.
  \item Otherwise, let $x_0(\rho)=(1-\rho)$ and
    \begin{align*}
      x_{k} =\left\{\begin{array}{ll}
          \rho^k(1-\rho) &\text{ if $\rho^{k+1}(1-\rho)>\gamma$}\\
          \frac{\rho^k(1-\rho) - \gamma}{1-\rho} &\text{ if $\rho^{k}(1-\rho)>\gamma>\rho^{k+1}(1-\rho)$}\\
          0&\text{otherwise}
        \end{array}\right.
    \end{align*}
  \end{itemize}
  
  It is straightforward to verify that for all $\rho$, $x(\rho)$ is a
  fixed point of the differential equation~\eqref{eq:DI}. Moreover, the
  quantity $\sum_{k}x_k(\rho)$ is an increasing function of $\rho$. This
  shows that the differential equation has a unique fixed point (for
  $\gamma<\gamma^*$).
  \qed
\end{proof}

\subsection{Impact on the Performance}

Theorem~\ref{th:DI} shows that the optimal redistribution rate
decreases with the station capacity. The optimal redistribution rate
is minimal for $s=K/2+\lambda/\mu$. In this case, moving bikes from
saturated to empty stations at rate $\lambda/(K-1)$ suffices to avoid the
existence of problematic stations. When
$x=\min(s-\lambda/\mu,K-s+\lambda/\mu)$ is an integer, the optimal
redistribution rate simplifies in $\gamma^*=\lambda/(2x-1)$.

\begin{figure}[ht]
  \centering
  \begin{tabular}{ccc}
    \subfigure[$s=5, \gamma=0$]{\includegraphics[width=.31\linewidth]{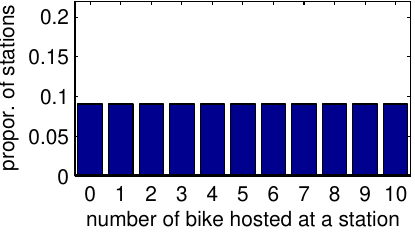}}&
    \subfigure[$s=5, \gamma=\gamma^*_5=1/9$]{\includegraphics[width=.31\linewidth]{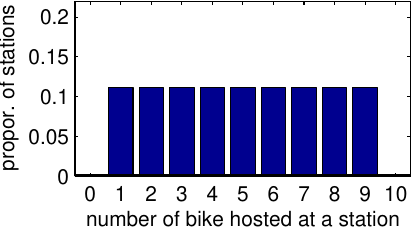}}&
    \subfigure[$s=5, \gamma=1/5$]{\includegraphics[width=.31\linewidth]{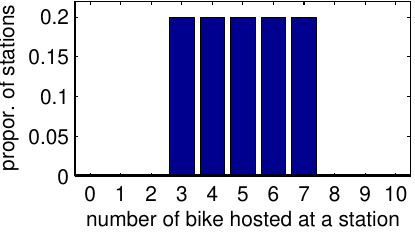}}\\
    \subfigure[$s=7, \gamma=0$]{\includegraphics[width=.31\linewidth]{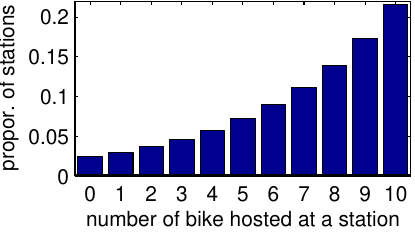}}&
    \subfigure[$s=7, \gamma=1/9$]{\includegraphics[width=.31\linewidth]{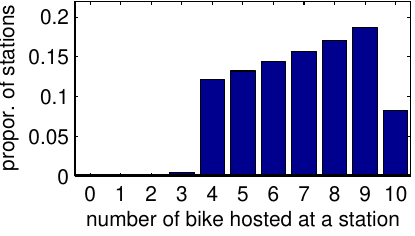}}&
    \subfigure[$s=7, \gamma=\gamma^*_7=1/5$]{\includegraphics[width=.31\linewidth]{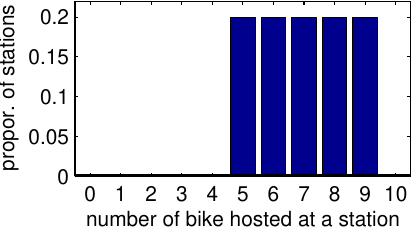}}
  \end{tabular}
  \caption{Illustration of Theorem~\ref{th:DI}: distribution of
    station occupancy. The station capacity is $K=10$. 
    Two fleet sizes $s=5$ and $s=7$ and three
    redistribution rates $\gamma=0$, $\gamma=1/9$ and $\gamma=1/5$ are compared.}
  \label{fig:occup_redistribution}
\end{figure}

These results are illustrated in
Figure~\ref{fig:occup_redistribution}.  The station capacity is set to
$K=10$ and $\mu=+\infty$ (the trip time is negligible). The proportion
of stations that have $x$ bikes available is plotted as a function of
$x$.  Two fleet sizes $s=K/2=5$ and $s=7$ are compared, according to
various values of $\gamma$: $\gamma=0$, $\gamma=\gamma^*_5=1/9$ (the
optimal redistribution rate for $s=K/2$) and $\gamma=\gamma^*_7=1/5$
(the optimal redistribution rate for $s=7$).

In both cases, the occupancy distribution concentrates around $s$ as
$\gamma$ increases.  When $s=5$, the occupancy distribution is uniform
for $\gamma=0$. As expected, there is no problematic station when
$\gamma\ge\gamma^*_5=1/9$.  When $s=7$, the occupancy distribution is
geometric for $\gamma=0$ (see Section~\ref{sec:basic-model}). For
$\gamma=1/9<\gamma^*_7$, there is no empty station and the occupancy
distribution is  truncated geometric. When $\gamma=\gamma^*_7$, the
occupancy distribution is uniform on $\{5\dots9\}$ and there is no
problematic station.


\section{Validation of the Model: Extensions and Simulations}
\label{sec:validation}

\subsection{Time- and Space-Inhomogeneous Systems}

In this paper, we focus on homogeneous bike-sharing systems. This
means that the travel demand is constant with time and that this time
is the same for any pair of origin and destination.  The model
presented in Section~\ref{sec:basic-model} captures the main features
of these systems, \emph{i.e.}, loss of the arriving users, search when
returning a bike. As we consider a homogeneous model, its performance
is naturally described in terms of the proportion of problematic
stations in steady-state.

A natural extension of these results is to consider
space-inhomogeneity and time-inhomogeneity. Space-inhomogeneity often
occurs in cities where some stations are geographically higher or
lower than others, thus creating a flow from one region of the city to
another. Our basic model can be directly extended to this case, for
example, by considering clusters of stations that have a similar level
of popularity. The steady-state behavior of such a model is exposed by
\cite{velib-aofa}. This enables us to compute the fleet size that is
optimal in terms of minimizing the proportion of problematic stations
in a given cluster. Because of working hours or week-ends,
bike-sharing systems are often time-inhomogeneous. Modeling these
phenomena can be done by considering tides of people that go from
housing to working areas in the morning and come back in the evening,
as in \cite{waserhole2013pricing}. In this case, the proportion of

problematic stations does not reflect the performance of the system
and the definition of a performance metric is not clear and might
depend on the situation. Characterizing and understanding such systems
is an issue beyond the scope of this paper, and we plan to tackle it
in future work.

\subsection{Distribution of Trip Times}
\subsubsection{One-Choice Model: Insensibility to the Distribution of Trip Times}

In order to obtain a tractable model, we choose the inter-arrival
times of users and trip times to be exponentially distributed. This
leads to a Markovian model that has a compact representation.  We
investigate the influence of more realistic distributions by using
simulation. Our preliminary simulation results indicate that if the
average trip time has an influence on the performance, the actual
distribution has little effect.  The behavior of a system with general
trip-time distribution is very similar to a system where trip times
have an exponential distribution with the same mean.

To verify this assumption, we compare four possible trip-time
distributions:
\begin{enumerate}
\item Original model -- The trip times are exponentially distributed
  of mean $1/\mu$.
\item Deterministic -- The trip times are all equal to $1/\mu$.
\item Log-normal -- The trip times follow a log-normal distribution
  of mean and standard deviation $1/\mu$.
\item Uniform -- The trip times are uniformly distributed on $[0;
  2/\mu]$.
\end{enumerate}
In all cases, the average trip time is $1/\mu$. We simulate the four
cases on a system composed of $N=100$ stations. Apart from the trip
distribution, the model is exactly the same as in Section~\ref{sec:hm}
(users arrive according to a Poisson process and there is no
geometry).  The results are reported in
Figure~\ref{fig:trip_distribution}.

\begin{figure}[ht]
  \centering
  \begin{tabular}{@{}cc@{}}
    \subfigure[\label{fig:lambdamu1} $\lambda/\mu=1$]{\includegraphics[width=0.46\linewidth]{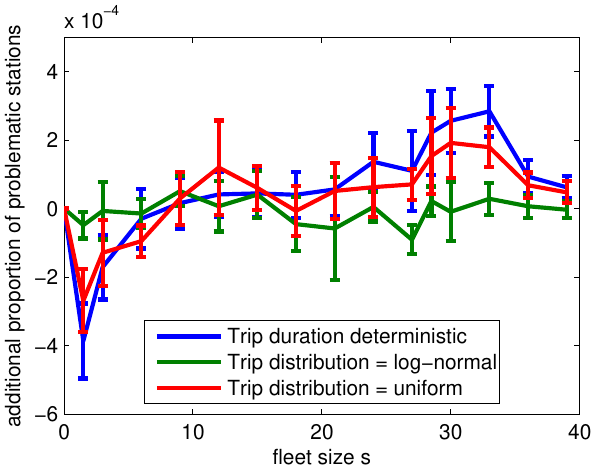}}&
    \subfigure[\label{fig:lambdamu10} $\lambda/\mu=10$]{\includegraphics[width=0.46\linewidth]{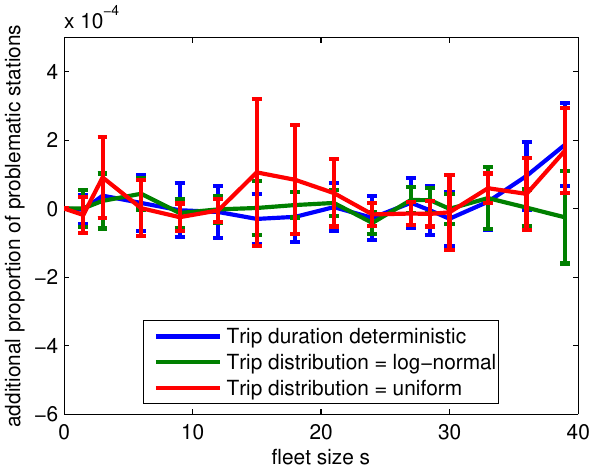}}
\end{tabular}
\caption{Difference between the proportion of problematic stations when
  the trip follows another distribution minus when the trip follows an
  exponential distribution.  The vertical bars indicate the 95\%
  confidence interval for the mean.   }
  \label{fig:trip_distribution}
\end{figure}


For each distribution, we plot the proportion of problematic stations
for this distribution minus the proportion of problematic stations for
the original model.  The results are shown as a function of the fleet
size $s$ for two situations: $\lambda=\mu$
(Figure~\ref{fig:lambdamu10}) and $\lambda=10\mu$
(Figure~\ref{fig:lambdamu10}).  The results reported in
Figure~\ref{fig:trip_distribution} are the average over 100
simulations. The errorbars are the 95\% confidence interval. We
observe that, for a fixed value of $\lambda/\mu$, the average
difference is always small and is almost smaller than the confidence
interval. However, we recall that the average trip duration $1/\mu$
does have an influence (see Figure~\ref{fig:performance_onecluster}
and Figure~\ref{fig:performance_s}).

We conclude that the trip distribution has a negligible effect on the
performance.  The performance of the system depends only on the
average trip time $1/\mu$. Note that we also simulate a case where
the trip time is exponentially distributed but when the time between
two arrivals of customers follows one of the four distributions. The
results, not reported here, show that the inter-arrival time
distribution also has a negligible effect on the proportion of
problematic stations.

\subsubsection{Two-Choice Model with Delays: Larger Trip Times Degrades the Performance}

In this section, we study by simulation the effect of the average trip
time on the performance of the two-choice policy with another
two-choice model described as follows. Note that such a model seems
analytically much more difficult.  It is an homogeneous model with
$N=100$ stations and $s=K/2+\lambda/\mu$ bikes per station. Users
arrive at rate $\lambda=1$ in each station and take a bike if the
station is not empty. When a user takes a bike, she picks two stations
at random and chooses the one that has the smallest number of bikes
availables as her destination. She arrives at her chosen destination
after a trip time with exponential distribution. If her destination
has no available spot, she tries another station at random. This model
has two differences with the two-choice model of
Section~\ref{sec:incentives}. First, we now assume that the choice of
the least loaded station is performed before the trip of the
users. Furthermore, in the new model, a user who cannot return her
bike picks a station at random and not the least loaded among two.

\begin{figure}[ht]
  \centering
  \includegraphics[width=\linewidth]{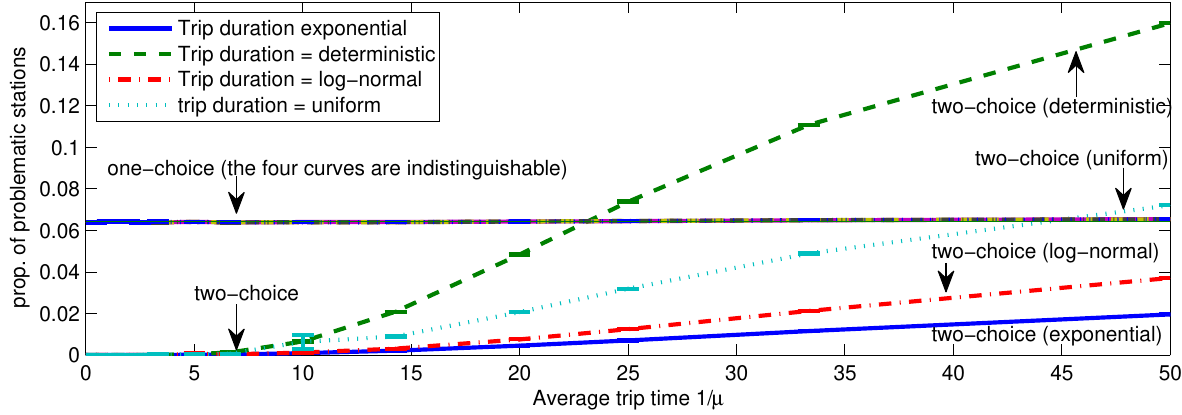}
  \caption{Two-choice model where the users choose their destination before riding their bike: 
    proportion of problematic stations as a
    function of the average trip duration.  The performance degrades as
    the average trip time increases.  We compare the four distribution
    of trip times.  For the one-choice model, the performance does not
    depend on the trip-time distribution.}
  \label{fig:trip_distribution_2choice}
\end{figure}

We simulate this model for four trip-time distributions (exponential,
log-normal, uniform and deterministic).  We plot in
Figure~\ref{fig:trip_distribution_2choice} the proportion of
problematic stations as a function of the average trip time,
$1/\mu$. In each case, the capacity of station is $K=30$ and the fleet
size is equal to $s=K/2+\lambda/\mu$.  We observe that, as expected,
the performance for the two-choice model degrades as the trip time
increases. This is explained by the delay in the information: when the
trip time is large, a station that had few bikes available when the
user departed does not necessarily still have few bikes when the user
arrives.
 The performance of the
original model is not affected by the variation of the average trip time, because the
fleet size is equal to $s=K/2+\lambda/\mu$.

We compare the two models (one-choice and two-choice).  For a very
large average trip time, the performance when the two-choice rule is
applied can even be worse than when it is not. Although
counterintuitive, this phenomena can be explained by the delays: for
example, let us consider that all trips have a duration of
$10$min. Then, a station that has few bikes available at time $t$ will
have many users who decide to ride to it between time $t$ and time
$t+10$min. These users will not arrive before time $t+10$min. As a
consequence, it is likely this station will experience a burst of
arrivals after time $t+10$min, causing more problems than if the
choices were made at random.  As shown in
Figure~\ref{fig:trip_distribution_2choice}, this phenomena is
exacerbated when the trip-time distribution is more concentrated (the
proportion for problematic stations is increasing when going for
deterministic to uniform to log-normal to exponential distribution).

We want to emphasize that this problem occurs only when the trip time
is very large compared to the arrival rate (about $50$ when the
distribution is log-normal).  Hence, we conjecture that the problem
would not occur in a realistic scenario. We plan to study this
question in a future work.

\subsection{Influence of Geometry on the Performance Metric}
\label{sec:geometry}



Our theoretical results and closed-form formula strongly rely on the
assumption that the system has no geometry.  In real-world systems,
finding or returning a bike can induce a local search for an available
bike or an available spot.  Studying  such systems analytically is out of
reach. This section presents simulation results that show that the
influence of geometry on the proportion of problematic station is
limited, both for the basic model and the two-choice model. However,
we show in Section~\ref{sec:other_perf} that it has an effect when
other metrics are considered, such as the time to return a bike. 


{\bf The model.} We consider two representations of geometry
represented in Figure~\ref{fig:geom_1choice} and
\ref{fig:geom_2choice}: a 2D grid and a single line.  The 2D grid is a
schematic representation of a homogeneous city center like that of
Manhattan. This situation aims at being a good representation of many
bike-sharing systems: the stations are placed quite evenly on a plane
and each station has a few neighbors (here, four) spread around
it. The 1D line is a more extreme case that corresponds to a city
spread along a single road. We expect the imbalances due to random
choices to have more effect in the 1D line than in the 2D grid.

\begin{figure}[ht]
  \centering
  \begin{tabular}{ccc}
    \begin{tabular}{c}
      \subfigure[\label{fig:2D_onechoice}$2D$ grid of $N$ stations.]{
        \begin{tikzpicture}
          \draw[step=.5,black] (.5,.5) grid (3,3); \draw (1.25,0.75)
          edge[bend left,->] (1.75,2.25); %
          \draw[dashed] (1.75,2.25) edge[->] (1.25,2.25); %
          \draw[dashed] (1.75,2.25) edge[->] (1.75,1.75); %
          \draw[dashed] (1.75,2.25) edge[->] (1.75,2.75); %
          \draw[dashed] (1.75,2.25) edge[->] (1.75,2.75); %
          \draw[dashed] (1.75,2.25) edge[->] (2.25,2.25); %
          \node at (2.2,2.4) {\tiny local search if saturated};
        \end{tikzpicture}
      }
      \\[-5pt]
      \subfigure[\label{fig:1D_onechoice}Line of $N$ stations. ]{
        \begin{tikzpicture}
          \draw[step=.5,black] (0,0) grid (3,.5); \draw (.75,0.5)
          edge[bend left,->] (2.25,.4);

          \draw (2.25,.4) edge[dashed,->] (2.75,.3); %
          \draw (2.25,.4) edge[dashed,->] (1.75,.3); %
          \node at (2,.2) {\tiny local search if saturated};
        \end{tikzpicture}
      }
    \end{tabular}
    &
    \begin{tabular}{c}
        \includegraphics[width=.47\linewidth]{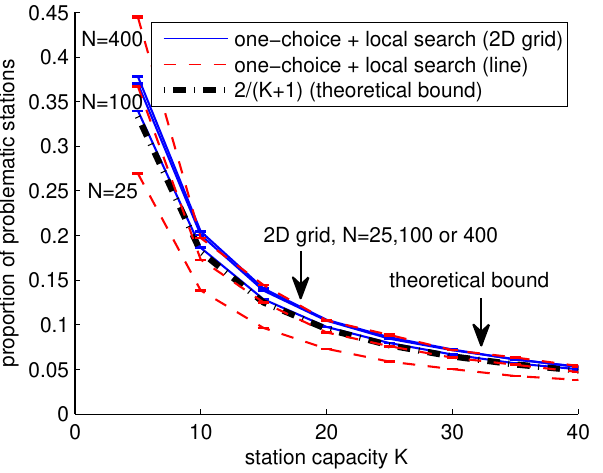}
    \end{tabular}
  \end{tabular}
  \caption{Influence of the local search for the one-choice
    model. Proportion of problematic stations as a function of the
    station's capacity $K$ for two geometric models: 2D grid and the
    line. The parameters are $s=K/2$ and $N=25$, $10$,
    $400$.}
  \label{fig:geom_1choice}
\end{figure}
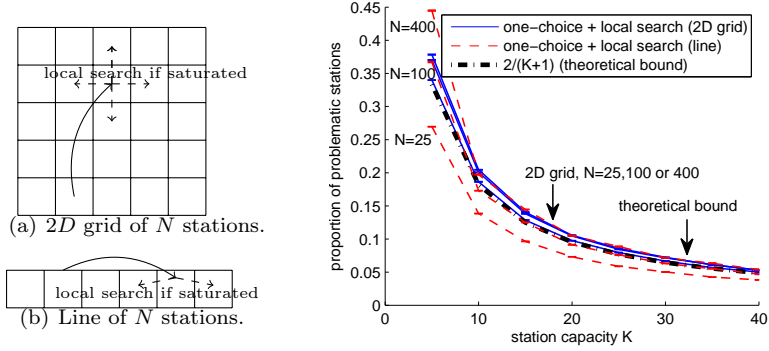

We simulate the two models. Users arrive at each station with rate
$\lambda$. If the station is empty, the user leaves the system. Otherwise,
she takes a bike. In the one-choice case (see
Figure~\ref{fig:geom_1choice}), she chooses a destination at random. If
this station is saturated, she performs a random walk on the neighbors of
the destination until she finds a non-saturated station. In the
two-choice case (see Figure~\ref{fig:geom_2choice}), the user chooses the
least loaded station among two neighbors. Again, the user performs a
local search if the station is saturated.

The proportion of problematic stations for all cases is reported in
Figure~\ref{fig:geom_1choice} and Figure~\ref{fig:geom_2choice} for
$\mu=+\infty$. In both cases, the models were simulated with $N=25$,
$N=100$ and $N=400$ stations. Each point represents the average over
$20$ independent simulations. The error bars indicate $95\%$
confidence intervals but are most of the time too small to be seen. We
compare these values with the theoretical bounds of the models without
geometry: $2/(K+1)$ for the one-choice model
(Figure~\ref{fig:geom_1choice}) and $\sqrt{K}2^{-K/2}$ for the
two-choice model (Figure~\ref{fig:geom_2choice}).  In all cases, the
performance exhibits the same trend in the models with geometry as the
theoretical bounds. In particular, the performance obtained for the 2D
grid are mostly independent of $N$ and are very similar to the
theoretical bounds. This shows that the bounds obtained in
Theorems~\ref{th:one-choice} and \ref{th:two-choice} are
representative of more realistic systems, even if they are obtained on
models that do not take into account the geometry.



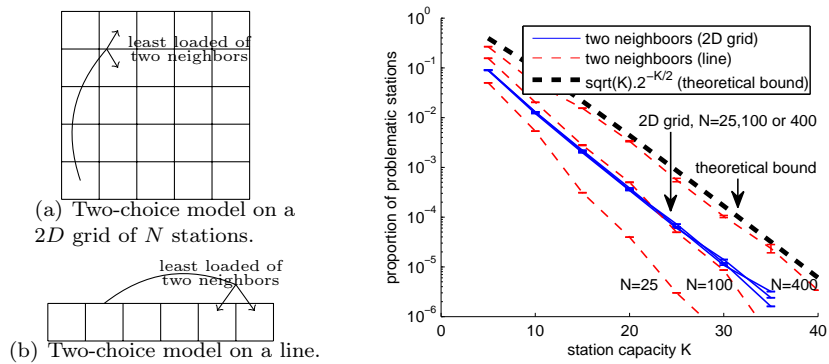
\begin{figure}[ht]
  \centering
  \begin{tabular}{cc}
    \begin{tabular}{c}
      \subfigure[\label{fig:2D_twochoice}Two-choice model on a $2D$ grid of $N$ stations.]{
        ~~\begin{tikzpicture}
          \draw[step=.5,black] (.5,.5) grid (3,3); \draw (.75,0.75) edge[bend
          left] (1.1,2.5); \draw (1.1,2.5) edge[->] (1.25,2.75); \draw
          (1.1,2.5) edge[->] (1.25,2.25); \node at (2.2,2.6) {\tiny least
            loaded of}; \node at (2.2,2.4) {\tiny two neighbors};
        \end{tikzpicture}
        ~~
      }
      \\[-8pt]
      \subfigure[\label{fig:1D_twochoice}Two-choice model on a line.]{
        ~~~
        \begin{tikzpicture}
          \draw[step=.5,black] (0,0) grid (3,.5); \draw (.75,0.5) edge[bend
          left] (2.5,.75); \draw (2.5,.75) edge[->] (2.75,.4); \draw
          (2.5,.75) edge[->] (2.25,.4); \node at (2.3,1.) {\tiny least
            loaded of}; \node at (2.3,0.8) {\tiny two neighbors};
        \end{tikzpicture}
      ~~~}
    \end{tabular}
    &
    \begin{tabular}{c}
      \includegraphics[width=.47\linewidth]{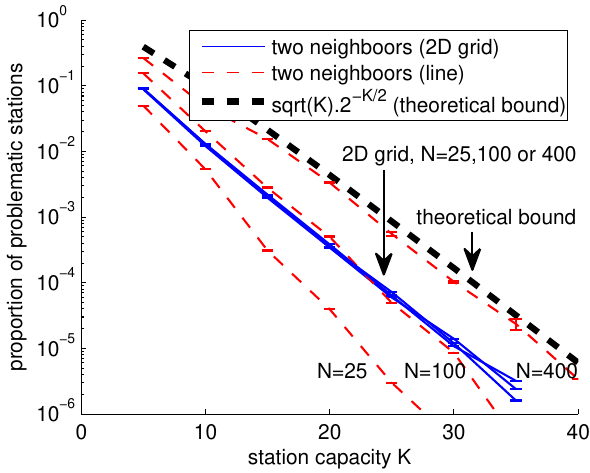}
    \end{tabular}
  \end{tabular}
  \caption{Influence of geometry on the two-choice model. Proportion
    of problematic stations as a function of the station's capacity
    $K$ for the 2D grid and the local search. }
  \label{fig:geom_2choice}
\end{figure}

\subsubsection{Average Number of Visited Stations}
\label{sec:other_perf}

We now consider a different performance indicator that is the
average number of stations that have to be visited before a bike can
be returned. This metric is an indicator of how much time is necessary
to return a bike. It is critical for users. 

In the original model without geometry, this metric can be easily
computed as a function of the proportion of saturated stations. Let
$p$ be this proportion. With probability $1-p$, the original
destination is not saturated and the user visit only one
station. Otherwise (with probability $1-p$) , the user chooses another
destination at random and repeats this operation until she finds a
non-saturated station. As the new destination is chosen at random, it
also has a probability $p$ to be saturated. Hence, the number of
visited stations before returning a bike is, on average,
$(1-p)\sum_{i=0}^\infty(1+i)p^i=1/(1-p)$.

In our models with geometry (line or 2D grid), we consider that users
perform a local search to return their bikes: if a destination is
saturated, the user chooses one of the two (or four) neighbors and
repeats the operation until she finds a available spot. This implies
that the neighbors of a saturated station are likely to receive more
bikes than stations not located nearby. Hence, the neighbors of a
saturated station are more likely to be saturated than others. This
creates local saturation and increase the average number of station
that a user has to visit before being able to return her bike.


\begin{figure}[ht]
  \centering 
  \begin{tabular}{@{}c@{~~}c@{}}
    \subfigure[\label{fig:geometry_s}Average number of
    visited stations before returning as a function of the fleet size $s$.]{%
      \includegraphics[width=0.48\linewidth]{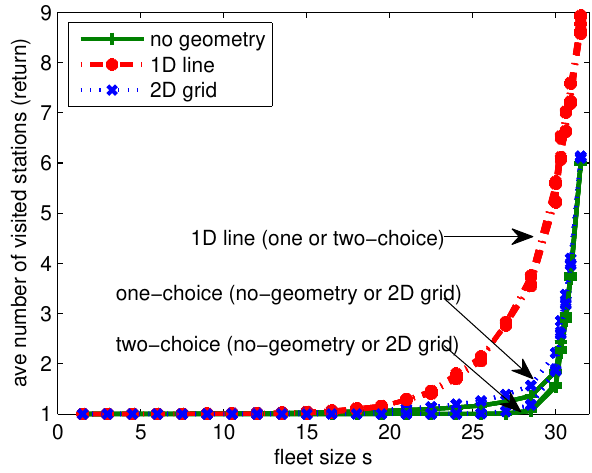}}&
    \subfigure[\label{fig:geometry_full}Average number of visited
    stations as a function of the proportion of saturated stations.]{%
      ~\includegraphics[width=0.48\linewidth]{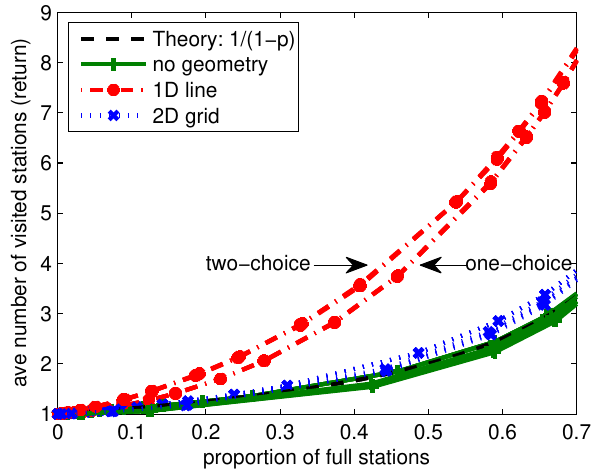}~~}
  \end{tabular}

  \caption{Average number of stations that have to be visited before
    returning a bike. We compare the geometric and non-geometric
    models for different fleet sizes and incentive strategies. Because
    of local search, the average number of visited stations grows much
    faster in the case of a line than for the 2D grid. The 2D
    grid case is similar to the model without geometry. }
  \label{fig:geometry2}
\end{figure}

We simulate the three models and compute the average number of
stations that a user visits before finding an available spot. The
results are reported in Figure~\ref{fig:geometry2}.  In
Figure~\ref{fig:geometry_s}, we plot the average number of visited
stations as a function of the fleet size $s$. This figure is to be
compared with Figure~\ref{fig:performance_onecluster} and
Figure~\ref{fig:performance_s} for the one-choice case and with
Figure~\ref{fig:power2} for the two-choice case. The capacity of the
station is $K=30$. When the number of bikes per station is low (fewer
than $15$), there are almost no saturated stations. In this case,
users successfully return their bike at their chosen destinations,
most of the time. When the number of bikes is close to $30$, we
observe that the average number of visited stations rises quickly in
the case of the line.  In all cases, the results are plotted for the
four trip-time distributions, but the curves are indistinguishable.
For the 2D model, the average number of visited stations is slightly
higher that than the non-geometrical case but remains similar.

To ease the comparison between the two metrics, we plot the average
number of visited station as a function of the proportion of saturated
station in Figure~\ref{fig:geometry_full}. As indicated before, when
the new destination is chosen at random, the average number of visited
stations is $1/(1-p)$; where $p$ is the proportion of saturated
stations. We observe in Figure~\ref{fig:geometry_full} that the 2D
grid has a similar behavior as the non-geometric model.  When the
geometry is represented by a line, the number of visited stations is
higher but it has the same order of magnitude. In particular, when
$30$ percent or less of the stations are saturated, the average number
of visited stations is lower than $2.5$, and lower than twice the one
without geometry. We remark that having $30$ percent of the stations
that are saturated corresponds to an extreme situation\footnote{In our
  simulation, each station can host up to 30 bikes. Having more than
  $30$ percent of stations that are saturated occurs when there is a
  fleet of more than 25 bikes per station for the line and more than
  30 bikes per stations for the 2D grid.}.

To conclude, our simulations show that a 2D grid exhibits a
performance similarly to the theory. As the positioning of stations is
similar to a 2D grid in many bike-sharing systems, this implies that
the basic model reflects the behavior of realistic scenarios.

\subsection{Adding Features to the Model (Search at Arrival, Shorter Search Time when Returning)}
\label{sec:local_search}

The features of the homogeneous model can also be changed, while
keeping the model analytically
tractable.  
For example, instead of leaving, a user who arrives at an empty
station can visit randomly other stations to find a bike, consisting
of attempts with exponential distribution with parameter
$\lambda'$. The loss of users or the avoidance of empty stations can
be seen as the two extreme cases $\lambda'=0$ and $\lambda'=+\infty$.
A similar modification can be done when returning a bike.  The mean
searching time $1/\mu'$ for finding or returning a bike can also be
different for the mean trip time $1/\mu$.

These modifications do not change the nature of the model: the
occupancy of station will still follow a geometric distribution. Only
the influence of the fleet size $s$ will change.  For example, if the
reattempt times of a user who cannot return her bike at the first
attempt, are exponentially distributed of parameter $\mu'$, then
Equation~\eqref{fixep} is replaced\footnote{ To see that, the
  proportion $s$ of bikes per station is the sum of two terms: the
  mean number of bikes per station $\sum_{k=0}^Kk\nu_\rho(k)$ and the
  mean number of users per station riding. This term is the product of
  the effective arrival rate $\lambda(1-\nu_\rho(0))$ times the mean
  riding time. The mean riding time is
  $1/\mu+(1-\nu_\rho(K))\sum_{k=0}^{+\infty} \nu_\rho(K)^k k/\mu'$
  because it is the sum of the mean trip time $1/\mu$ and $k/\mu'$ if
  the user returns at the $(k+1)$-th attempt, i.e. with probability
  $(1-\nu_\rho(K))\nu_\rho(K)^k$. It leads to the result.

 
} by
\begin{equation}
  \label{eq:mu'}
  s=\frac{\lambda}{\mu}\rho\Big[1+\nu_\rho(K)\Big(\frac{\mu}{\mu'}-1\Big)\Big] + \sum_{k=0}^Kk\nu_\rho(k).
\end{equation} 

\begin{figure}[ht]
  \centering
  \begin{tabular}{@{}cc@{}}
    \subfigure[\label{fig:performance_s1}$\lambda/\mu=1$]{%
      \includegraphics[width=.48\linewidth]{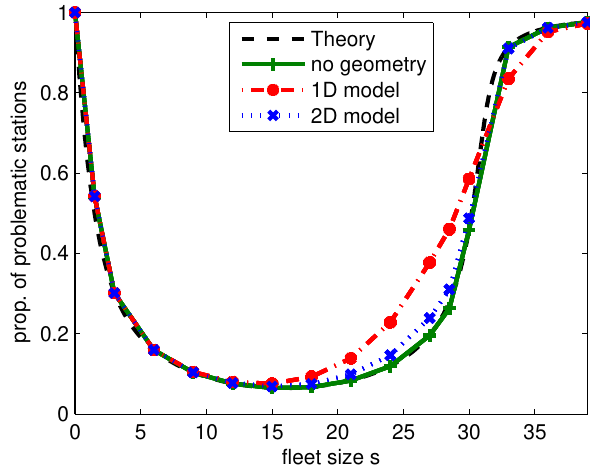}}&
    \subfigure[\label{fig:performance_s10}$\lambda/\mu=10$]{%
      \includegraphics[width=.48\linewidth]{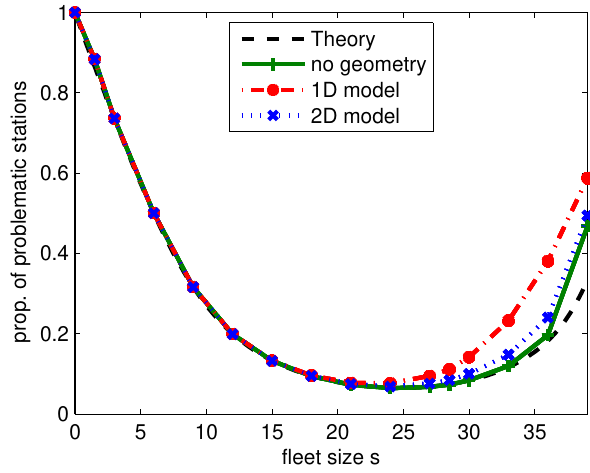}}
  \end{tabular}
  \caption{Proportion of problematic stations as a function of the
    fleet size $s$ when the search around a saturated destination are $5$
    times faster that trip time ($\mu'=5\mu$). The results are ploted
    for the four trip-time distributions but the curves are
    indistinguishable (they are within $0.15\%$). }
  \label{fig:performance_s}
\end{figure}

Moreover we simulate this model and report the proportion of problematic
stations as a function of the fleet size per station $s$ in
Figure~\ref{fig:performance_s}. We plot the results for
$\lambda/\mu=1$ (Figure~\ref{fig:performance_s1}) and $\lambda/\mu=10$
(Figure~\ref{fig:performance_s10}). In each case, we compare the
theoretical prediction of Equation~\eqref{eq:mu'} with a simulation
for $N=100$ and the two  models with geometry presented in
Section~\ref{sec:geometry}.  We set $\mu'=5\mu$ and we vary and
compare the four trip-time distribution (exponential, deterministic,
log-normal and uniform).  In all cases, the trip-time distribution has
a negligible effect on the performance: for a given configuration, the
relative difference between the proportions of problematic stations is at most $0.15\%$. 
The geometry does have an effect, but the overall
behavior is similar. As mentionned in Section~\ref{sec:geometry}, the
performance of the 2D model is similar to the performance of a model
without geometry.

\subsection{Influence of Truck Capacity}
\label{sec:truck_capacity}

In the redistribution model presented in Section~\ref{sec:regulation},
we assume that bikes are moved individually by a truck. This leads to
a simple formula for the optimal redistribution rate. In practice,
however, bikes are moved by trucks that can contain a few tens of
bikes.  This section reports simulation results of the model described
in Section~\ref{sec:geometry} where a truck of capacity $C$ is
added. To obtain a fair comparison, the rate at which the truck visits
the stations is set inversely proportional to the truck
capacity. Hence, with this scaling, having a larger truck capacity
leads to a poorer performance: the balance achieved when bikes are
moved one by one at rate $10$ is better than when bikes are moved ten
by ten at rate $1$.

The simulated model is composed of $N$ stations that are placed in a
2D grid, as in Figure~\ref{fig:2D_onechoice}. Users move as in the
one-choice model of Section~\ref{sec:geometry}: At each time step, a
user arrives at one station picked at random, takes a bike if this
station is not empty, and performs a local search if the targeted
destination is saturated.  At each time step, with probability $\gamma/C$,
a truck transports bikes from the station that has the largest number
of bikes and places them in the station that has the smallest number
of bikes.  The truck tries to equalize the number of bikes between the
two stations but cannot move more than $C$ bikes at a time. The case
$C=1$ corresponds to the model described before. The maximum number of
bikes per time slot that can be moved by truck does not depend on $C$
and is equal to $C(\gamma/C)=\gamma$.

\begin{figure}[ht]
  \centering
  \begin{tabular}{cc}
    \subfigure[$K=15$]{\includegraphics[width=.45\linewidth]{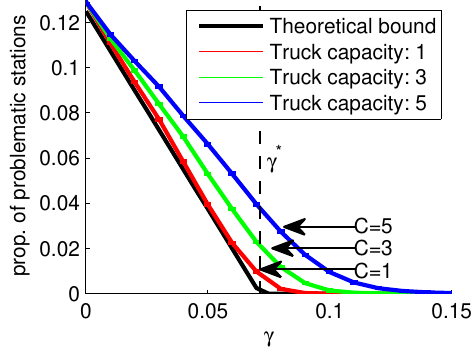}}&
    \subfigure[$K=30$]{\includegraphics[width=.45\linewidth]{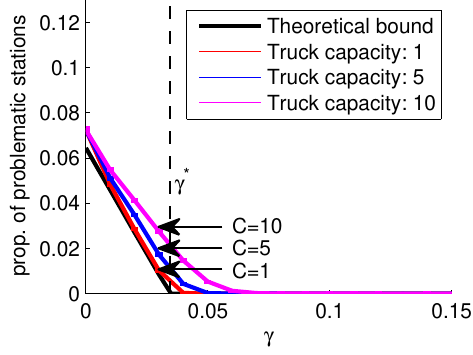}}
  \end{tabular}\vspace{-.2cm}
  \caption{Proportion of problematic stations as a function of the
    redistribution rate $\gamma/\lambda$. We compare the theoretical
    analysis with simulations for various truck  capacities.  The
    theoretically optimal redistribution rate is $\gamma^*=1/(K-1)$.}
  \label{fig:perf_regulation}\vspace{-.3cm}
\end{figure}

The proportion of problematic stations are reported in
Figure~\ref{fig:perf_regulation} for two stations capacity: $K=15$ and
$K=30$. Recall $s=K/2$. For $K=15$, simulation results for $C=1$,
$C=3$ and $C=5$ are compared with the theoretical values obtained from
the fixed point analysis. For $K=30$, simulation results for $C=1$,
$C=5$ and $C=10$ are compared with the theoretical values.  The
vertical lines represent the optimal redistribution rate $\gamma^*$,
obtained from Theorem~\ref{th:DI}.  We observe that the theoretical
model (with a truck capacity of one) predicts qualitatively the
performance of the simulated models. This prediction is an optimistic
estimation of the simulated values. Moreover, as expected, the
performance decreases with the truck capacity.

To conclude, in an homogeneous system, the optimal redistribution rate
depends on the capacity of the station and leads to a great
improvement of the system performance. It can be shown that 
combining this redistribution mechanism with the two-choice incentives
introduced in the previous section leads to an optimal redistribution
rate close to $O(\sqrt{K}2^{-K/2})$.



\section{Conclusion and Future Work}
\label{sec:conclusion}

In this paper, we investigate the influence of the station capacities
on the performance of homogeneous bike-sharing systems.  Using a
stochastic model and a fluid approximation, we provide analytical
expressions for the performance.  They are summarized in
Table~\ref{tab:capacity}. The optimal fleet size is approximately
$K/2$ for all models.  Without using incentives, the capacity has only
a linear effect on the performance or on the optimal redistribution
rate.  For this purpose, an incentive to return bikes to the least
loaded station among two improves dramatically the performance, even
if a small proportion of users accept to do this.  Moreover, even if
this model does not take into account any geographic aspect of the
system, simulations show that these results also hold when considering
simple geometric models with local interactions.

\begin{table}[htb]
  \centering
  \begin{tabular}{|c|c|c|}
    \hline
    &\begin{tabular}{c}
      Minimal proportion\\of problematic stations
      \end{tabular}
      &Optimal fleet size $s$\\\hline
    Original model & $2/(K+1)$ &$s=K/2+\lambda/\mu$ \\\hline
    Two-choice & $\sqrt{K}2^{-K/2}$ & $s-\lambda/\mu\in[K/2;K-\log_2(K)]$\\\hline
    Regulation& 0 if $\gamma\ge\lambda/(K-1)$ & $s=K/2+\lambda/\mu$%
    \\\hline
  \end{tabular}
  \caption{Summary of the main results: influence of the station capacity $K$ on the proportion of problematic stations.  }\vspace{-.2cm}
  \label{tab:capacity}
\end{table}

Our results prove that the effect of random choices on the performance
should not be neglected when studying the performance of a
bike-sharing system. Even in a completely balanced system, they
dramatically affect the performance.  A natural extension of this work
is to consider stations with different parameters.  The steady-state
performance of such a system is given by \cite{velib-aofa}. It proves
that, without repositioning via incentives or trucks, performance is
very poor. One interesting question is whether the steady-state
performance can be used as a metric in a system with varying operation
conditions, such as peak-hours and non-peak hours.  Our work can serve
as a building block for studying the effect of incentives and
redistribution mechanisms. Studying practical implementations of these
mechanisms in real-world systems is postponed for future work.
Moreover, the transient behavior of such mechanisms in a city where
the attractiveness of stations varies over time could be studied.

\bibliographystyle{abbrv}
\bibliography{ref}

\ifRapportRecherche

\fi

\end{document}